\documentclass{article}
\usepackage{fullpage}

\usepackage{graphicx} 

\usepackage{amsmath}
\usepackage{amssymb}
\usepackage{amsthm}
\usepackage{appendix}
\usepackage[numbers]{natbib}
\usepackage[noend,ruled,noline,linesnumbered]{algorithm2e}
\usepackage{xcolor}
\usepackage{url}

\usepackage{longfbox}
\usepackage{multirow}
\usepackage{makecell}
 \usepackage{enumitem}
 \usepackage{setspace}

\SetCommentSty{mycommfont}
\DontPrintSemicolon
\usepackage{soul}

\newtheorem*{theorem*}{Theorem}
\newtheorem{theorem}{Theorem}

\newtheorem{lemma}[theorem]{Lemma}

\newtheorem{observation}{Observation}
\newtheorem{corollary}[theorem]{Corollary}

\usepackage{thm-restate}

\newcommand{\mech}{\textsc{CFA}}
\newcommand{\mechname}{Constant Fit with Advice~}

\newcommand{\arbset}{\mathbb{Y}}
\newcommand{\PP}{\mathbb{P}}
\newcommand{\lmech}{\textsc{LFA}}
\newcommand{\lmechname}{Linear Fit with Advice~}
\newcommand{\binmech}{\textsc{RCA}}
\newcommand{\binmechname}{Random Classification with Advice~}

\title{Strategyproof Learning with Advice}
\author{%
  Eric Balkanski \\
  Columbia University \\
  \texttt{eb3224@columbia.edu}
    \and
      Cherlin Zhu \\
  Columbia University \\
  \texttt{cz2740@columbia.edu}
  }
\date{}

\begin{document}

\maketitle
\begin{abstract}
An important challenge in robust machine learning is when training data is provided by strategic sources who may intentionally report erroneous data for their own benefit. A line of work at the intersection of machine learning and mechanism design aims to deter strategic agents from reporting erroneous training data by designing learning algorithms that are strategyproof. Strategyproofness is a strong and desirable property, but it comes at a cost in the approximation ratio of even simple risk minimization problems. 
    
    In this paper, we study strategyproof regression and classification problems in a model with advice. This model is part of a recent line on mechanism design with advice where the goal is to achieve both an improved approximation ratio when the advice is correct (consistency) and a bounded approximation ratio when the advice is incorrect (robustness). We provide the first non-trivial consistency-robustness tradeoffs for strategyproof regression and classification, which hold for simple yet interesting classes of functions. For classes of constant functions, we give a deterministic and strategyproof  mechanism that is, for any $\gamma \in (0, 2]$, $1+\gamma$ consistent and $1 + 4/\gamma$ robust and provide a lower bound that shows that this tradeoff is optimal. We extend this mechanism and its guarantees to homogeneous linear regression over $\mathbb{R}$. In the binary classification problem of selecting from three or more labelings, we present strong impossibility results for both deterministic and randomized mechanism. Finally, we provide deterministic and randomized mechanisms for selecting from two labelings.

\end{abstract}

\pagenumbering{gobble}
\newpage

\pagenumbering{arabic}
\section{Introduction}

In a growing number of machine learning applications, training data is provided by agents who have preferences over the output of the machine learning system. A commonly used example is retail distribution, where regression is employed to optimize inventory shipping. For large retailers such as Zara, part of the training data is provided by store managers who report their predicted demand for items~\citep{caro2010inventory,caro2010zara}. Since managers have financial incentives to increase sales in their own stores, they may strategically manipulate the predicted demand that is reported to the central warehouses. In particular, there is evidence that this manipulation has occurred when there is a limited supply of a high-selling item~\citep{caro2010zara}. The potential errors in the training data due to such strategic manipulation are, of course, an issue for the robustness of machine learning models.

In order to avoid such manipulations, a line of work at the intersection of machine learning and mechanism design has developed  learning algorithms  that are robust to strategic data sources~\citep{dekel2010incentive,perote2004strategy,chen2018strategyproof,cummings2015truthful,meir2011strategyproof,meir2008strategyproof,meir2010limits,meir2012algorithms,hardt2016strategic,ghalme2021strategic,dong2018strategic,ahmadi2021strategic}. These  algorithms, also called mechanisms, deter strategic data manipulations by satisfying a strategyproofness property that guarantees that these agents cannot benefit from reporting erroneous data, regardless of the reports of the other agents. Strategyproofness is a strong and desirable property, but even for simple risk minimization problems, it comes at a cost in the best achievable approximation ratio. 

The existing work on strategyproof regression and strategyproof classification assumes that the learning mechanism has no information about the agents' data besides what these agents report. This assumption is often  pessimistic.  In the retail distribution example, it is reasonable to assume that central warehouses have access to, in addition to the future demands reported by store managers, rough estimates of what these future demands should be. 

A recent line of work on mechanism design with advice has shown that such side information can be leveraged to overcome worst-case bounds in mechanism design. Strategyproof mechanisms that achieve an improved approximation ratio when the advice is accurate (consistency) and an acceptable approximation ratio when the advice is inaccurate (robustness) have been designed for problems  such as strategic facility location~\citep{agrawal2022learning,xu2022mechanism,balkanski2024randomized},   auction design~\citep{lu2023competitive, caragiannis2024randomized,balkanski2023online}, strategic scheduling~\citep{balkanski2022strategyproof}, strategic assignment~\citep{colini2024trust}, and metric distortion~\citep{berger2023optimal}. In this paper, we aim to leverage side advice to overcome pessimistic worst-case impossibility results in strategyproof learning.

\paragraph{The model.} In strategic learning, each agent $i \in \{1, \ldots, n\}$ reports a set $S_i = \{(x_{i,j}, y_{i,j})\}_j$ of labeled data points to the learner.  The points $x_{i,j}$ are public information, but the labels $y_{i,j}$ are private information that agent $i$ can potentially misreport. The goal of the mechanism is to learn a function $f$ from a function class $\mathcal{F}$ that minimizes the global risk  $$R(f, S) = \frac{1}{|S|} \sum_{i=1}^n \sum_{j = 1}^{|S_i|} \ell(f(x_{i,j}), y_{i,j})$$
for some loss function $\ell$. The agents are strategic and aim to minimize their personal risk $R_i(f, S) = \frac{1}{|S_i|}\sum_{j = 1}^{|S_i|} \ell(f(x_{i,j}), y_{i,j}).$ 

We augment the problem of strategic learning with a potentially erroneous advice $\tilde{f} \in \mathcal{F}$ about the global risk minimizer that is given as input to the mechanism, in addition to the labeled data reported by the agents. We note that this advice is weaker than the advice
$\tilde{S} = \cup_{i=1}^n \{(x_{i,j}, \tilde{y}_{i,j})\}_j$ about the labeled points since, for the learning problems that we consider, the function $\tilde{f}$ that minimizes the global risk over $\tilde{S}$ can be efficiently computed when given $\tilde{S}$. 

Given this advice $\tilde{f}$, we say a mechanism is $\alpha$ consistent if it achieves an $\alpha$ approximation when the advice is equal to the optimal, i.e., risk minimizing, function, and we say a mechanism is $\beta$ robust if it achieves a $\beta$ approximation given any arbitrarily bad advice. Additional details about the model are provided in Section~\ref{sec:prelims}.

\paragraph{Our results.} We first consider  regression problems with the absolute loss function $\ell(x,y) = |x- y|$ in Section~\ref{sec:regression}. For the class of  constant functions $\mathcal{F}_{\mathbb{R}}$, \citet{dekel2010incentive} gave a deterministic and group-strategyproof mechanism that is a $3$ approximation to the minimum global risk and showed that no deterministic strategyproof mechanism can achieve an approximation ratio better than $3$. Our first main  result  is a tight consistency-robustness tradeoff for the class $\mathcal{F}_{\mathbb{R}_{\leq T}}$  of constant functions upper bounded by a sufficiently and arbitrarily large $T$.

\begin{theorem*}
    For any $\mathbb{Y} \subseteq \mathbb{R}$ and $\gamma \in (0,2]$, there is a deterministic and strategyproof mechanism for the class  of constant functions $\mathcal{F}_{\mathbb{Y}}  = \{f_a(x) = a : a \in \mathbb{Y}\}$ that is $1 + \gamma $ consistent and $1 + 4/\gamma$ robust. Moreover, for any $T$ sufficiently large and $\gamma \in (0,2]$, any  deterministic, strategyproof, and $1 + \gamma $ consistent mechanism for the class  of constant functions $\mathcal{F}_{\mathbb{R}_{\leq T}}$ is $1 + 4/\gamma - \Omega(1/n)$ robust.
\end{theorem*}

\renewcommand{\arraystretch}{1.35}
\begin{table} \centering \label{t:summary}
\begin{tabular}{|c|l  l || c | c  c |}
    \hline
    & \multirow{2}{*}{\makecell{Function \\ class}} & & \multirow{2}{*}{\makecell{No advice}} & \multicolumn{2}{ c |}{With advice}\\
    & & &  & Consistency & Robustness \\\hline \hline
    \multirow{2}{*}{Regression} & Constant & det. & $3^\star$  & $1+\gamma$ & $1 + \frac{4}{\gamma}$ 
    \\ \cline{2-6}
        & \makecell[l]{Homogeneous \\ linear} & det. & $3^*$ & $1+\gamma$ & $1+\frac{4}{\gamma}$ \\
     \hline
       \multirow{4}{*}{Classification}  & \multirow{2}{*}{\makecell[l]{General \\ labelings}}  & det. & $O(n)$\textsuperscript{†}  &if $o(n)$ & must be unbounded 
       \\
    & & rand. & $3$\textsuperscript{†}   & $3 - \frac{2}{n}$ & $\Omega(n)$ 
    \\ \cline{2-6}
     & \multirow{2}{*}{\makecell[l]{Two \\ labelings}} & det. & $3$\textsuperscript{†}  & $1+\gamma$ & $1+\frac{4}{\gamma}$ 
     \\
    & & rand. & $2$\textsuperscript{†} & $1 + \gamma$ & $1+\frac{1}{\gamma}$ 
    \\\hline
\end{tabular}
\caption{\small{Summary of results. The results are for $\gamma \in (0,2]$ except random classification over two labelings, which is for $\gamma \in (0,1]$.  Results labeled with $\star$ are from \citet{dekel2010incentive}, and those labeled with † are from \citet{meir2012algorithms}.}}
\end{table}

  Note that when  $\gamma$ is a small constant, the mechanism achieves near-optimal consistency and constant robustness. With $\gamma = 2$, the mechanism is $3$ consistent and $3$ robust, which recovers the optimal approximation  without advice. In fact, in this case with $\gamma = 2$, our mechanism is exactly the mechanism of  \citet{dekel2010incentive}. A main technical challenge is the construction of the family of hard instances and its analysis. We also show that the approximation achieved by our mechanism smoothly interpolates between its consistency and robustness as a function of the advice error.  This mechanism and its guarantees also extend to  homogeneous linear functions over $\mathbb{R}$. 

We then consider binary classification problems over the $0$-$1$ loss function in the shared input setting, where agents have identical points but may disagree on their labels, and the function class is a set of specific labelings for the points (Section~\ref{sec:classification}).  \citet{meir2012algorithms} give a deterministic  mechanism that achieves a $2n - 1$ approximation and show that no deterministic mechanism can achieve a sublinear approximation. They also give a $3$ approximate randomized mechanism, which they also show to be tight. Our second main result is a strong impossibility result for classification with advice. 

\begin{theorem*}
    For binary classification in the shared input setting, any deterministic and $o(n)$ consistent mechanism  has unbounded robustness. Any random mechanism with consistency better than $3$ is $\Omega(n)$ robust.
\end{theorem*}

These results imply that no robustness guarantee is achievable by  deterministic mechanisms with a consistency that improves over the best-known approximation in the setting without advice. For randomized mechanisms, the robustness must be sacrificed to be at least linear, instead of $3$, to obtain an improved consistency. For the special case of function classes with only two labelings,  we extend the deterministic mechanism for regression and its guarantees, as well as provide a randomized mechanism parametrized by $\gamma \in (0,1]$ that is $1+\gamma$ consistent and $1 + 1/\gamma$ robust. 

Finally, in Section~\ref{sec:learning}, we consider the learning-theoretic setting where the points are drawn from a distribution. We aim for risk and strategyproofness guarantees that hold in expectation over the distribution (as in standard learning problems) and show that our results generalize.

In summary (see Table~\ref{t:summary}), our results indicate that  side advice is helpful for strategic regression but that, in general, it is not helpful for strategic classification if we wish to maintain robustness guarantees.

\paragraph{Related work.}

\citet{perote2004strategy} and \citet{dekel2010incentive} initiated the study of  learning algorithms over data that is reported by strategic agents. Both papers provide regression learning mechanisms on data where agents can misreport the labels $y$ of data points. The former characterizes strategyproof linear regression mechanisms, which~\citet{chen2018strategyproof} extend to higher dimensions. The latter considers the empirical risk minimization problem over constant and homogeneous linear functions. \citet{procaccia2013approximate} study randomized mechanisms over constant functions. \citet{cummings2015truthful} consider the related problem of linear regression with potentially misreported features $x$ of data points, where the agents' behavior is governed by privacy concerns.  We note that strategic learning over constant functions  is related to  strategic facility location~\citep{procaccia2013approximate} and strategic geometric median~\citep{el2023strategyproofness}, with the main difference being that agents in strategic learning can have multiple points of interest.
 
 With regards to strategic classification,  a series of works  considers the binary classification problem where the mechanism designer must select from a finite set of labelings~\citep{meir2008strategyproof, meir2010limits, meir2011strategyproof, meir2011tight,
 meir2012algorithms}.    As in regression learning, there are also works that consider agents that strategically report the features $x$ of their data points~\citep{hardt2016strategic, ghalme2021strategic, ahmadi2021strategic,dong2018strategic}. 

The framework of algorithms with advice (also called algorithms with predictions or learning-augmented algorithms) formalized by \citet{lykouris2021competitive} has been applied, separately, to both learning problems and mechanism design 
 problems. Learning problems that have been studied in this framework include online classification~\citep{raman2024online} and $k$-means clustering~\citep{ergun2021learning,gamlath2022approximate,nguyen2022improved}. Mechanism design problems that have been studied with predictions include strategic facility location~~\citep{agrawal2022learning,xu2022mechanism,balkanski2024randomized, istrate2022mechanism, barak2024mac},  auction design~\citep{lu2023competitive, balkanski2023online, caragiannis2024randomized}, strategic scheduling~\citep{xu2022mechanism, balkanski2022strategyproof}, strategic assignment~\citep{colini2024trust},  metric distortion~\citep{berger2023optimal}, decentralized mechanism design~\citep{gkatzelis2022improved}, and  social welfare-revenue bicriteria mechanism design~\citep{prasad2024bicriteria}.


\section{Preliminaries}\label{sec:prelims}

In strategic learning, there are  $n$ agents and each agent $i \in [n]$ has data $S_i$ comprised of set $\{(x_{i,j},y_{i,j})\}_{j \in [|S_i|]}$ for points $x_{i,1}, \dots, x_{i,|S_i|} \in \mathcal{X}$ with labels $y_{i,1}, \dots, y_{i,|S_i|} \in \mathcal{Y}$,  where $\mathcal{X}$ and $\mathcal{Y} \subseteq \mathbb{R}$ are the input and output spaces. The points $x_{i,j}$ are public information, but the labels $y_{i,j}$ are private information to agent $i$. The full dataset is the multiset $S = \uplus_{i=1}^n S_i$, where $\uplus$ maintains copies of the same labeled point. A function class $\mathcal{F}$, or hypothesis space, consists of functions $f: \mathcal{X} \to \mathcal{Y}$. Given $S$ and $\mathcal{F}$, the goal of the learning algorithm (also referred to as mechanism) is to output a function $f \in \mathcal{F}$ that minimizes, for a given loss function $\ell: \mathbb{R}\times \mathbb{R} \to \mathbb{R}$, the global risk 
\[R(f, S) = \frac{1}{|S|}\sum_{i=1}^n\sum_{j=1}^{|S_i|}\ell(f(x_{i,j}), y_{i,j}).\]
We use the absolute loss for regression tasks, so $\ell(f(x),y) = |f(x) - y|$, and 0-1 loss for binary classification tasks, so for $f: \mathcal{X} \to \{0,1\}$ and $y \in \{0,1\}$ we have that $\ell(f(x),y) = \mathbf{1}(f(x) \neq y)$.
       
\paragraph{Strategyproofness.}
The personal risk  of agent $i$, which is defined as
$$R_i(f,S_i) = \frac{1}{|S_i|}\sum_{j=1}^{|S_i|}\ell(f(x_{i,j}), y_{i,j}),$$
 is their negative utility
 and each agent aims to minimize their personal risk. Each agent strategically reports $\bar{S}_i = \{(x_{i,j}, \bar{y}_{i,j})\}_{j \in [|S_i|]}$ where   $\bar{y}_{i,j} \in \mathcal{Y}$ are, potentially wrong, labels for each point $x_{i,j}$. 
 
 For an arbitrary $\varepsilon > 0$, a mechanism is $\varepsilon$ strategyproof with respect to $R_i$ if for every true dataset $\uplus_{i = 1}^nS_i$ and reported data $\bar{S}_i$ of agent $i$, we have that $R_i(f, S_i) - \varepsilon \leq R_i(\bar{f}, S_i)$ where $f$ and $\bar{f}$ are the functions returned by the mechanism given $\uplus_{i = 1}^nS_i$ and $\bar{S}_i\uplus(\uplus_{j \neq i}S_j)$, respectively. In other words, agent $i$ cannot decrease their personal risk by more than $\varepsilon$ by misreporting incorrect data, regardless of the data reported by the other agents. Similarly, a mechanism is $\varepsilon$-group-strategyproof if, for every set $C \subseteq [n]$, every true data $\uplus_{i = 1}^n S_i$, and every potential report from the coalition $\uplus_{i \in C}\bar{S}_i$, we have that if $R_i(\bar{f}, S_i) \leq R_i(f, S_i) - \varepsilon$ for all $i \in C$, then $R_i(\bar{f}, S_i) = R_i(f, S_i) - \varepsilon$ for all $i \in C$, where $f$ and $\bar{f}$ are the functions returned by the mechanism when given $\uplus_{i = 1}^nS_i$ and $(\uplus_{i \in C}\bar{S}_i) \uplus( \uplus_{i \notin C}S_i)$, respectively. In other words, no coalition of agents can both (a) all decrease their risk by at least $\varepsilon$ and (b) have at least one member who decreases their risk by more than $\varepsilon$. We refer to $0$-strategyproof and $0$-group-strategyproof as strategyproof and group-strategyproof.

\paragraph{Strategic learning with advice.} We augment the strategic learning problem with an advice $\Tilde{f} \in \mathcal{F}$  regarding the optimal function.  Let $\mathcal{M}(S, \Tilde{f})$ denote the function output by a mechanism  $\mathcal{M}$ given dataset $S$ and advice $\Tilde{f}$. Mechanism $\mathcal{M}$ is $\alpha$ consistent if for all $S$, if $f^* \in \arg\min_{f \in \mathcal{F}}R(f, S)$, then $\frac{R(\mathcal{M}(S,f^*), S)}{R(f^*, S)} \leq \alpha.$ It is $\beta$ robust if for all datasets $S$ and advice $\Tilde{f} \in \mathcal{F}$, $\frac{R(\mathcal{M}(S, \Tilde{f}), S)}{R(f^*, S)} \leq \beta.$

\section{Strategyproof regression with advice}\label{sec:regression}

In this section, we study regression problems with advice. We primarily focus on the class of constant functions and extend our results to  homogeneous linear functions. We present the mechanism with advice  and its analysis in Section~\ref{subsec:PFA}, the lower bound that shows that the consistency-robustness achieved by the mechanism is tight in Section~\ref{subsec:PFAneg}, and the extension to homogeneous linear functions in Section~\ref{subsec:PFAlinear}. Missing proofs are in Appendix~\ref{appendix:regression}.

\subsection{The \mechname  mechanism} \label{subsec:PFA}

We consider for any $\arbset \subseteq \mathbb{R}$ the class of constant functions $\mathcal{F}_{\arbset} = \{f_a : a \in \arbset\}$ where $f_a(x) = a$. For ease of presentation, we abuse notation and denote $f_a$ as $a$ and drop the input values $x$ in dataset $S$. We assume that $\mathcal{Y} = \arbset$, which can be interpreted as restricting agents to reporting their preferences amongst possible outcomes. Thus, we consider a class of functions $\mathcal{F}_{\arbset}$, for an arbitrary $\arbset \subseteq \mathbb{R}$, where the input to the mechanism is an advice $\Tilde{a} \in \arbset$ and, for each agent $i$, a set  $S_i = \{y_{ij}\}_j$ of reported labels,  with  $y_{ij} \in \arbset$, and its output is a function $a \in \arbset$.

\paragraph{Description of the mechanism.} 
Let $\textsc{ERM}(\arbset, S) = \arg\min_{a \in \arbset} R(a, S)$ 
be a global risk minimizer over $\mathcal{F}_\arbset$ 
with respect to $S$. For constant functions and the absolute loss $\ell(a, y) = |a - y|$, it is easy to see that $\textsc{ERM}(\arbset, S)$ is the median of $S$, where we assume that ties are broken arbitrarily but consistently for even $|S|$. The mechanism also uses a risk function $R_\lambda(a, S, \Tilde{a})$ defined as
   \[R_\lambda(a, S, \Tilde{a}) = \frac{1}{(1 + \lambda)|S|}\left(\sum_{i=1}^n\sum_{j = 1}^{|S_i|}|a - y_{i,j}| + \lambda\sum_{i=1}^n\sum_{j=1}^{|S_i|}|a - \Tilde{a}|\right).\]
Note that if $\lambda|S|$ is an integer and we denote by $\{a\}^z$ the multiset containing $m$ copies of an arbitrary point $a$, then $R_\lambda(a, S, \Tilde{a}) = R(a, S \uplus \{\tilde{a}\}^{\lambda|S|}).$ Thus $R_\lambda(a, S, \Tilde{a})$ is the risk with respect to $S$ augmented with $\lambda |S|$ copies of $\Tilde{a}$, where a fractional number of copies is allowed. The mechanism, called \mechname (\mech) and formally defined below as Mechanism~\ref{mech:pfa}, first projects each agent's input $S_i$ to the constant $b_i = \textsc{ERM}(\arbset, S_i)$ that minimizes their personal risk, which leads to a new set of labels  $\{b_i\}^{|S_i|}$. It then returns the global risk minimizer $\arg\min_{a \in \arbset}R_{\frac{2 - \gamma}{2 + \gamma}}(a, \uplus_{i = 1}^n\{b_i\}^{|S_i|}, \Tilde{a})$ on the new dataset $\uplus_{i = 1}^n\{b_i\}^{|S_i|}$ augmented with $\frac{2 - \gamma}{2 + \gamma} |S|$ copies of the advice $\tilde{a}$. The copies of $\tilde{a}$ thus push the output function towards the advice. We assume that ties are broken by selecting the largest constant function. Parameter $\gamma \in [0,2]$ represents the confidence in the advice: selecting $\gamma$ closer to $0$ results in heavier reliance on the advice.

\vspace{.2cm}

\begin{algorithm}[H]
\setstretch{1.1}
\SetKwInOut{Input}{Input}
\Input{set of values $\arbset$, parameter $\gamma \in [0,2]$, agents' labels $S = \uplus_{i=1}^n S_i$, advice $\Tilde{a} \in \arbset$}
\For{$i \in [n]$}
{$b_i \leftarrow \textsc{ERM}(\arbset, S_i)$}
\Return $\arg\min_{a \in \arbset}R_{\frac{2 - \gamma}{2 + \gamma}}(a, \uplus_{i = 1}^n\{b_i\}^{|S_i|}, \Tilde{a})$ 
\caption{\mechname (\mech)}
\label{mech:pfa}
\end{algorithm}

\vspace{.2cm}

We denote the output of this mechanism as $\mech_{\gamma, \arbset }(S, \Tilde{a})$. 

\paragraph{The analysis of the mechanism.}
The main result for our mechanism is the following.
\begin{theorem}\label{thm:robustconsist}
    For any $\gamma \in (0, 2]$ and $\arbset \subseteq \mathbb{R}$
    , the $\mech_{\gamma, \arbset }$ mechanism is strategyproof,  $1+4/\gamma$ robust, and $1 + \gamma$ consistent. 
\end{theorem}

 Thus, our result holds not only for the class $\mathcal{F}_{\mathbb{R}}$ of all constant functions but also for any class that is a subset of the class of all constant functions. We first focus on analyzing the consistency and robustness. We start by providing a helper lemma that allows us to express the approximation of selecting constant $a$ on instance $S$ in terms of the number of points in $S$ at least $a$. 

\begin{restatable}{rLem}{lemRatio}\label{lem:ratio}
    Consider instance $S$ and $\arbset \subseteq \mathbb{R}$. Let $a^* \in \arg\min_{a \in \arbset }R(a, S)$ and $a \in \arbset $. If $a \geq a^*$, then with $b = |\{y_{i,j} \in S: y_{i,j} \geq a\}|$, we have that  $\frac{R(a, S)}{R(a^*, S)} \leq \frac{|S| - b}{b}$. If $a \leq a^*$, then with $b = |\{y_{i,j} \in S: y_{i,j} \leq a\}|$, we also have that $\frac{R(a, S)}{R(a^*, S)} \leq \frac{|S| - b}{b}$.
\end{restatable}

\begin{proof} Assume that $a \geq a^*$ and let $b = |\{y_{i,j} \in S: y_{i,j} \geq a\}|$. Let $d = a - a^*$, then we can decompose the empirical risk of picking $a^*$ as follows:

    \begin{align*}
        R(a^*, S) & = \frac{1}{|S|}\left[\sum_{i,j:y_{i,j}< a^*}(a^*-y_{i,j}) + \sum_{i,j:a^* \leq y_{i,j} < a}(y_{i,j} - a^*) + \sum_{i,j:y_{i,j} \geq a}(y_{i,j}-a^*)\right] \\
        & = \frac{1}{|S|}\left[\sum_{i,j:y_{i,j} < a^*}(a^* - y_{i,j}) + \sum_{i,j:a^* \leq y_{i,j}< a}(y_{i,j} - a^*) + \sum_{i,j:y_{i,j}\geq a}(d + y_{i,j} - a)]\right] \\
        & \geq \frac{1}{|S|}\left[\sum_{i,j:y_{i,j} < a^*}(a^*-y_{i,j}) + \sum_{i,j:y_{i,j} \geq a}(y_{i,j}-a) + d|\{y_{i,j} \in S:y_{i,j}\geq a\}|\right].
    \end{align*}
We can similarly decompose $R(a, S)$:
    \begin{align*}
        R(a, S) & = \frac{1}{|S|}\left[\sum_{i,j:y_{i,j}< a^*}( a-y_{i,j}) + \sum_{i,j:a^* \leq y_{i,j} <  a}( a - y_{i,j}) + \sum_{i,j:y_{i,j} \geq  a}(y_{i,j}- a)\right] \\
        & = \frac{1}{|S|}\left[\sum_{i,j:y_{i,j} < a^*}(d + a^* - y_{i,j}) + \sum_{i,j:a^* \leq y_{i,j}<  a}( a - y_{i,j}) + \sum_{i,j:y_{i,j}\geq  a}(y_{i,j} -  a)\right] \\
        & \leq \frac{1}{|S|}\left[\sum_{i,j:y_{i,j} < a^*}(a^*-y_{i,j}) + \sum_{i,j:y_{i,j} \geq  a}(y_{i,j}- a) + d|\{y_{i,j} \in S:y_{i,j}< a\}|\right].
    \end{align*}

    The last line is because for $y_{i,j} \in [a^*, a)$, $a - y_{i,j} \leq a -a^*$. Then we get the following bound on the approximation ratio:
    \begin{align*}
    \frac{R(a, S)}{R(a^*, S)} & \leq \frac{\sum\limits_{i,j:y_{i,j}< a^*}(a^* - y_{i,j}) + \sum\limits_{i,j: y_{i,j} \geq  a}(y_{i,j}-  a) + d|\{y_{i,j} \in S:y_{i,j} <  a\}|}{\sum\limits_{i,j:y_{i,j}< a^*}(a^* - y_{i,j}) + \sum\limits_{i,j: y_{i,j} \geq  a}(y_{i,j}- a) + d|\{y_{i,j} \in S:y_{i,j} \geq  a\}|} \\
    & \leq \frac{|\{y_{i,j} \in S :y_{i,j}< a\}|}{|\{y_{i,j} \in S: y_{i,j} \geq  a\}|} = \frac{|S| - b}{b}.
    \end{align*}
The last inequality holds because $a$ is at least $a^*$, the median of $S$, and therefore $|\{y_{i,j} \in S:y_{i,j} < a\}| \geq |\{y_{i,j} \in S : y_{i,j} \geq a\}|$.
    
The case  $a \leq a^*$ follows by an identical argument but with $b = |\{y_{i,j} \in S: y_{i,j} \leq a\}|$.
\end{proof}

 We also use the following observations regarding the \mech~mechanism. For each agent $i$ define data $S'_i = \{b_i\}^{|S_i|}$ with $b_i$ as defined in Mechanism~\ref{mech:pfa}, and dataset $S' = \uplus_{i=1}^nS_i'$. The first observation relates the number of points above some constant $a$ in $S$ to that in the constructed set $S'$ in \mech.

\begin{observation}\label{obs:halfpoints}
    For any set $S$ and $a \in \mathbb{R}$,  $|\{y_{i,j} \in S : y_{i,j} \geq a\}| \geq \frac{1}{2}|\{y'_{i,j} \in S': y'_{i,j} \geq a\}|$.
\end{observation}
This holds because each $y'_{i,j} = b_i$ is a median of the original $y_{i,j}$ values for $i$. Let $\lambda = (2-\gamma) / (2 + \gamma)$. Then we can express the output of \mech~as the risk minimizer over set $S'\uplus \{\Tilde{a}\}^{\lambda|S|}$.
\begin{observation}\label{obs:mechmed} For any advice $\Tilde{a}$,  the $\mech_{\gamma, \arbset }(S, \Tilde{a})$ mechanism  returns $a$ that satisfies $|y'_{i,j} \in S'  : y'_{i,j} \geq a| + \lambda|S| \cdot \mathbf{1}_{\Tilde{a} \geq a} \geq (1 + \lambda)|S|/2$ and $|y'_{i,j} \in S'  : y'_{i,j} \leq a| + \lambda|S| \cdot \mathbf{1}_{\Tilde{a} \leq a} \geq (1 + \lambda)|S|/2$.
\end{observation} 
This observation shows that, amongst the union of projected points $S'$ and $\gamma |S|$ copies of $\tilde{a}$, at least half of these points fall on either side of the constant returned by PFA. This holds by the equivalence $\lambda \sum_{i=1}^n\sum_{j =1}^{|S_i|}|a - \Tilde{a}| = \lambda \cdot |S|\cdot |a - \Tilde{a}|$ which equates $R_\lambda(a, S', \Tilde{a})$ to the risk of choosing $a$ over new set $S' \uplus \{\Tilde{a}\}^{\lambda \cdot|S|}$, of which \mech~returns a minimizer. Using these two observations, let us first upper bound consistency.

\begin{restatable}{rLem}{lemmaConsistency} \label{lem:PFAconsistency}
    For any $\gamma \in (0, 2]$ and $\arbset \subseteq \mathbb{R}$, $\mech_{\gamma, \arbset }$ achieves $1 + \gamma$ consistency.
\end{restatable}

\begin{proof}
Let $a^* \in \arbset $ be an optimal solution and assume that we have accurate advice $\Tilde{a} = a^*$. Denote the solution to $\mech_{\gamma, \arbset }(S, \Tilde{a})$ as $a$. If $a = a^*$ then our mechanism is optimal on this instance. Assume otherwise and consider the case where $ a > a^*$ (the symmetric argument also holds). We want to lower bound $|\{y_{i,j} \in S: y_{i,j} \geq  a\}|$. For simplicity, we assume that $\lambda \cdot |S| \in \mathbb{N}$. Take $S'$ as calculated in \mech. Then
\begin{align*}|\{y_{i,j} \in S: y_{i,j} \geq  a\}|  \geq \frac{1}{2}|\{y'_{i,j} \in S': y'_{i,j} \geq  a\}| \geq \frac{1}{2}\left(\frac{|S| + \lambda \cdot |S|}{2}\right) \geq \frac{(1 + \lambda)|S|}{4}.
\end{align*}
where the first inequality is from Observation~\ref{obs:halfpoints} and the second inequality is a consequence of Observation~\ref{obs:mechmed} and of the values at least 
$a$, none can be copies of $a^*$, and therefore all correspond to some $y'_{i,j}$.  We  apply this to Lemma~\ref{lem:ratio} to get
\[\frac{R(a, S)}{R(a^*, S)} \leq \frac{|S| - |\{y_{i,j}\in S: y_{i,j} \geq  a\}|}{|\{y_{i,j} \in S: y_{i,j} \geq  a\}|} \leq \frac{3 - \lambda}{1 + \lambda} = 1 + \gamma. \qedhere\] 
\end{proof}

The proof for the robustness guarantee follows a similar structure as the proof of Lemma~\ref{lem:PFAconsistency} for consistency.

\begin{restatable}{rLem}{lemPFARobust}\label{lem:PFArobust}
        For any $\gamma \in (0, 2]$ and $\arbset \subseteq \mathbb{R}$, $\mech_{\gamma, \arbset }$ achieves $1+4/
        \gamma$ robustness.
\end{restatable}

\begin{proof}
Consider an arbitrary input dataset $S$ and let $a^* \in \arbset$ be an optimal solution. We assume that the advice $\Tilde{a}$ is such that  $\Tilde{a} \geq a^*$. If $\Tilde{a} \leq a^*$, the claim follows from a symmetric argument. Let $a$ be the constant returned by $\mech_{\gamma, \arbset}(S, \Tilde{a})$. Consider the case where $a < a^*$. Let us denote $b$ as the solution returned by $\mech_{2, \arbset}(S, \Tilde{a})$, which we observe means that the advice is not used in the mechanism. It is clear that $a$ is between $b$ and $\Tilde{a}$ by Observation~\ref{obs:mechmed}, and thus $b \leq a < a^*$. We know by Lemma~\ref{lem:PFAconsistency} that selecting $b$ achieves $1+\gamma = 3$ consistency, and since the advice is not used, it achieves this 3 approximation for all $\Tilde{a}$. Thus we know that  constant $a$ achieves approximation at most 3 since $R(a, S)$ is nonincreasing over $[b, a^*]$, and $1+\frac{4}{\gamma} \geq 3$ for all $\gamma \in (0,2]$ so we still achieve $1 + \frac{4}{\gamma}$ robustness for this case.

Then assume that $a \geq a^*$. We now want to lower bound $|\{y_{i,j} \in S: y_{i,j} \geq  a\}|$. Take $S'$ as calculated in \mech. Then 
\begin{align*}|\{y_{i,j} \in S: y_{i,j} \geq  a\}|  \geq \frac{1}{2}|\{y'_{i,j} \in S': y'_{i,j} \geq  a\}|
 \geq \frac{1}{2}\left(\frac{(1+ \lambda) |S|}{2} - \lambda \cdot |S|\right) 
 = \frac{(1 - \lambda)|S|}{4},
\end{align*}
where the first inequality is from Observation~\ref{obs:halfpoints} and the second inequality is a consequence of Observation~\ref{obs:mechmed}. Since $a \geq a^*$, we can apply  Lemma~\ref{lem:ratio} to get
\[\frac{R(a, S)}{R(a^*, S)} \leq \frac{|S| - |\{y_{i,j} \in S: y_{i,j} \geq  a\}|}{|\{y_{i,j} \in S: y_{i,j} \geq  a\}|} \leq \frac{3 + \lambda}{1 - \lambda}.\]
Finally, we  substitute in $\lambda = \frac{2 - \gamma}{2 + \gamma}$ to get that this equals $1+4/\gamma.$
\end{proof}

We prove that \mech~is strategyproof, and when $|S_i|$ is odd for all $i$, the mechanism also satisfies the stronger guarantee of group-strategyproofness. When considering the learning-theoretic setting in Section~\ref{sec:learning}, which is an important practical motivation for our mechanism, we will discuss why odd $|S_i|$ is a reasonable assumption.

\begin{restatable}{rLem}{thmSP}\label{lem:PFASP}
    For any $\gamma \in [0, 2]$ and $\arbset \subseteq \mathbb{R}$, $\mech_{\gamma, \arbset }$ is strategyproof. Furthermore, $\mech_{\gamma,\arbset}$ is group-strategyproof on any instance where $|S_i|$ is odd for all $i$.
\end{restatable}
\begin{proof}
First, let us prove strategyproofness. Fix true instance $S = \{S_i, S_{-i}\}$ for some $i \in [n]$ and advice $\Tilde{a} \in \arbset$. If $\mech_{\gamma, \arbset}(S, \Tilde{a})$ returns $a_i^*$ which is a minimizer of $R_i(a, S_i)$, then clearly there is no incentive to misreport. Assume otherwise, so $a_i^* \notin \arg\min_{a \in \arbset}R_i(a, S_i)$. Observe that $\arg\min_{a \in \arbset}R_i(a, S_i) = \{a_i \in \arbset: a_i \in [b_i, c_i]\}$ for some $b_i \leq c_i$, representing the median(s) of $S_i$. WLOG let $a_i^* < b_i$. Recall from Observation~\ref{obs:mechmed} that \mech~returns some point $a$ such that $|y'_{i,j} \in S'  : y'_{i,j} \geq a| + \lambda|S|
\cdot \mathbf{1}_{\Tilde{a} \geq a} \geq (1 + \lambda)|S|/2$. Then agent $i$ can only change this solution if it reports $S_i$ such that $S_i'$ is composed of copies of $a_i' < a^*$. Doing so can only move the solution further from $a^*$ in the negative direction. It is well-known that $R_i$ for $S_i$ with respect to a constant increases as the constant gets further (in either direction) from the set of medians of $S_i$, so this misreport may only increase the risk of agent $i$, not decrease. Thus it is impossible for agent $i$ to decrease their risk through misreport.

     Now we prove group-strategyproofness when each $|S_i|$ is odd. We first want to reduce the space of values that each agent $i$ cares about from $|S_i|$ points to one. Observe that as defined in \mech, $S'_i$ is simply a projection of each agent $i$'s type to $|S_i|$ copies of the constant that minimizes their personal risk, and the agents can fully control this value. Then, we can equivalently consider that each agent reports $a_i$ as their preferred constant, and each $S'_i$ is precisely $|S_i|$ copies of $a_i$. It is, therefore, sufficient to prove that no group of agents can misreport their preferred constant function and all (weakly) decrease their personal risk with respect to the output of the mechanism.

     Let us assume that $\mech_{\gamma, \arbset}$ is not group-strategyproof, so that we have for some $\gamma \in (0, 2]$ a true instance $S$ with corresponding misreport $\bar{S}$ and advice $\Tilde{a}$ on which a coalition of agents $C \subseteq [n]$ can manipulate the mechanism in their favor. We can denote the mapping from their misreport as data $\bar{S}_i'$, from calculating $S_i'$ with the misreports in \mech, such that $\bar{S}_i' = S'_i$ for all $i \notin C$. Let $a$ and $\bar{a}$ be the constants returned by \mech~ given $S'$ and $\bar{S}'$ respectively, and WLOG let $a < \bar{a}$ (the symmetric argument holds otherwise). Then it must hold that $R_i(\bar{a}, S_i) \leq R_i(a, S_i)$ for all $i \in C$.

     Recall that in the mechanism we define $\lambda = \frac{2 - \gamma}{2 + \gamma}$. If we denote $T = S' \uplus \{\Tilde{a}\}^{\lambda \cdot |S|}$ and analogously $\bar{T} = \bar{S}' \uplus \{\Tilde{a}\}^{\lambda\cdot |S|}$, by Observation~\ref{obs:mechmed} we know that $a$ is a median of $T$ and $\bar{a}$ a median of $\bar{T}$. Let us introduce for real number $b$ and set of real numbers $Z$ sets $L_{b, Z}:=\{y \in Z: y \leq b\}$ and $R_{b, Z}:= \{y \in Z: y \geq b\}$, intuitively corresponding to points to the left and right of $b$. Applying this definition to medians $a$ and $\bar{a}$, we thus know that $|L_{a, T}|, |R_{\bar{a}, \bar{T}}| \geq \frac{(1 + \lambda)|S|}{2}$. 
     
     It is clear that for $S_i$ with odd size $|S_i|$, $S_i$ has a unique median. It is also well known that $R_i$ with respect to a constant increases as the constant gets further (in either direction) from the median of $S_i$. Then observe that for some agent $i$ with true preferred function $a_i$ such that $a_i \leq a < \bar{a}$, it holds that $R_i(a, S_i) < R_i(\bar{a}, S_i)$. Therefore, agent $i$ cannot be in the coalition, as he does not gain from the collective misreport. Thus, all of the points in $L_{a, T}$  come from either the reports of agents not in the coalition or copies of $\Tilde{a}$ and, therefore, must also be in $\bar{T}$. Then $|L_{a, \bar{T}}| \geq\frac{(1 + \lambda)|S|}{2}$. Observe that for $|R_{\bar{a}, \bar{T}}| \geq \frac{(1 + \lambda)|S|}{2}$ to also hold, we must have $|L_{a, \bar{T}}| = \frac{(1 + \lambda)|S|}{2}$, and therefore $|L_{a, T}| = \frac{(1 + \lambda)|S|}{2}$. This means that $(1 + \lambda)|S|$ must be even. However, since we tiebreak by selecting the largest median, this is a contradiction because there must exist point $\hat{a} > a$ such that $|R_{\hat{a}, T}| = \frac{(1 + \lambda)|S|}{2}$, and so $a$ cannot be the output of the mechanism with respect to $S$.
\end{proof}

Theorem~\ref{thm:robustconsist} then follows from Lemma~\ref{lem:PFAconsistency}, Lemma~\ref{lem:PFArobust}, and Lemma~\ref{lem:PFASP}. 

We conclude this section with an analysis of \mech~in terms of the \emph{error} of the advice. Let us introduce error term $\eta = \eta(S, \Tilde{a})$, which for any instance $S$ with advice $\Tilde{a}$ and set of optimal solutions $F_{opt} = \arg\min_{a \in \arbset }R(a, S)$ is defined as
$\eta = \min_{a^* \in F_{opt}}\frac{|a^* - \Tilde{a}|}{R(a^*, S)}.$ Then we get the following approximation that is a smooth transition from the consistency to the robustness guarantee as $\eta$ increases. 

\begin{restatable}{rThm}{thmError}
\label{thm:errortolerant}
    For any $\gamma \in (0, 2]$,  $\mech_{\gamma, \arbset }$ achieves $\min\left\{ 1+\frac{4}{\gamma}, 1 + \gamma + \eta\right\}$ approximation.
\end{restatable}

Let us first prove this simple helper lemma.

\begin{lemma}\label{lem:medianshift}
    For any instance $S$, parameter $\gamma \in (0, 2]$, and constants $a, a' \in \arbset$, if $\mech_{\gamma, \arbset}(S, a) = b$ and $\mech_{\gamma, \mathcal{F}_{\arbset}}(S, a') = b'$, then 
    \[|b'- b| \leq |a' - a|.\]
\end{lemma}

\begin{proof}
    WLOG we let $a \leq a'$. Recall from Observation~\ref{obs:mechmed} that $|y'_{i,j} \in S'  : y'_{i,j} \leq b| + \lambda|S| \cdot \mathbf{1}_{a \leq b} \geq (1 + \lambda)|S|/2$ of these points are at most $b$, which we call left-side points. When performing $\mech_{\gamma, \arbset}(S, a')$, we shift $\lambda \cdot |S|$ of the weight from $a$ to $a'$, which increases at most $\frac{(1+\lambda)|S|}{2}$ of the left-side points by $|a - a'|$, as $\lambda \cdot |S| \leq \frac{(1 + \lambda)|S|}{2}$ for $\lambda \in [0,1]$. Then the median, and consequently $b'$, can only increase from $b$ by at most $|a' - a|$, and thus $|b' - b| \leq |a' - a|$.
\end{proof}

We can now directly bound the approximation of \mech~as a function of $\eta$.

\begin{proof}[Proof of Theorem~\ref{thm:errortolerant}]
Fix $\gamma \in (0, 2]$. We know the upperbound of $1+\frac{4}{\gamma}$ approximation holds from robustness, so it suffices to prove that approximation is bounded above by $1 + \gamma + \eta$. Consider any instance $S$ with $n$ agents and $|S_i|$ points each. Let $\Tilde{a}$ be the given advice and $a^*$ be the optimal constant that minimizes $|a^* - \Tilde{a}|$, and let us denote $\mech_{\gamma, \arbset}(S, \Tilde{a}) = \Tilde{b}$ and $\mech_{\gamma, \arbset}(S, a^*) = b^*$. Then we can bound the risk incurred by $\Tilde{b}$:
    \begin{align*}
        R(\Tilde{b}, S) & = \frac{1}{|S|}\sum_{i=1}^n\sum_{j=1}^{|S_i|}|\Tilde{b} - y_{i,j}| \\
         & \leq \frac{1}{|S|}\sum_{i=1}^n\sum_{j=1}^{|S_i|}(|b^* - y_{i,j}| + |\Tilde{b} - b^*|) = R(b^*, S) + |\Tilde{b} - b^*| \\
         & \leq R(b^*, S) + |\Tilde{a} - a^*| \\
         & \leq (1 + \gamma)\cdot R(a^*, S) + |\Tilde{a} - a^*| = (1+\gamma+ \eta)R(a^*, S).
    \end{align*}
The inequalities respectively come from the triangle inequality, Lemma~\ref{lem:medianshift}, and the consistency of \mech. 
    \end{proof}

\subsection{Lower bound on consistency-robustness tradeoff} \label{subsec:PFAneg}



In this section, we prove a lower bound on the consistency-robustness tradeoff of strategyproof mechanisms over $\mathcal{F}_{\mathbb{R}_{\leq T}}$, for $T$ arbitrarily and sufficiently large,  that matches the performance of our \mech~mechanism, indicating that we achieve a tight tradeoff. This lower bound is a main technical contribution of the paper.

\begin{theorem}\label{thm:rationalconstantlb}
    Let $\epsilon > 0$. Then, for any  rational $\gamma \in (0,2]$ and $T > \frac{2}{\varepsilon} + \frac{8}{\gamma \varepsilon}$, there is no deterministic strategyproof mechanism over outcome space  $\mathcal{F}_{\mathbb{R}_{\leq T}}$ that is $1+\gamma$ consistent and $1 + 4/ \gamma - \epsilon$ robust.
\end{theorem}


Let us first introduce the family of instances we will leverage throughout this proof. 
For ease of presentation, let $Z_t(z,z') = \{\underbrace{z,\ldots,z}_t,\underbrace{z',\ldots,z'}_{t+1}\}$ for $t \in \mathbb{N}$ and $z,z' \in \mathbb{R}$, and $Z_t(z,z')^k$ be $k \in \mathbb{N}$ copies of $Z_t(z,z')$. 
To prove the lower bound, we first show for $k$ carefully chosen as a function of $\gamma$, that returning a constant at least $1$ given the binary instance $\{Z_t(0,0)^k\uplus Z_t(0,1)^{n-k}\}$
and correct advice $0$ violates $1+\gamma$ consistency. We then show that, as we increase all the 1 values incrementally to arbitrarily large $D$, a mechanism must still return a constant less than 1 to not violate strategyproofness. Finally, we show that, since our mechanism must return a constant less than $1$ on instance $\{Z_t(0,0)^k\uplus Z_t(0,D)^{n-k}\}$ with correct advice $0$, it must continue returning a constant less than 1 given (now incorrect) advice $0$ and instance $\{Z_t(D,0)^{k} \uplus\{Z_t(D,D)^{n-k}\}$ to satisfy strategyproofness. This results in robustness arbitrarily close to $1 + 4/\gamma$ as $n$ and $t$ grow. 

We will need some helper lemmas to prove this result. First, we want to be able to easily compare the risks incurred by selecting different constants for the same instance,  which we can do with this well-known property of the sum of absolute differences over a set.

\begin{observation} \label{obs:singlepeak}
    For any set $S$, if $a$ is the unique optimal constant function over $S$, then for any $b, c$ such that $a < b < c$ or $a > b > c$, it must hold $R(a, S) < R(b, S) < R(c, S)$. Similarly, for any $S_i \in S$ with unique optimal constant $a$, it must hold that $R_i(a, S_i) < R_i(b, S_i) < R_i(c, S_i)$.
\end{observation}

For ease of presentation, let $S(n, k, t, z)=\{Z_t(0,0)^k\uplus Z_t(0,z)^{n-k}\}$ for $n, k, t \in \mathbb{N}$, $k \leq n$, and $z \in \mathbb{R}$. We will now present some lemmas that will help us restrict the behavior of mechanisms on this family of instances.

\begin{restatable}{rLem}{lemZeroOpt}\label{lem:zeroopt}
       For any $n, k, t \in \mathbb{N}$ such that $0 < k \leq n$ and $t \geq n$, the optimal constant for instance $S(n, k, t, z)$ for any $z \in \mathbb{R}$ is $0$. 
\end{restatable}

\begin{proof}
    It is sufficient to show that for any such instance, more than half of the entries are 0, making 0 the median of the instance. To do so, we can show that there are more 0 entries than $z$ entries:
    \begin{align*}
        |\{y \in S: y = 0\}| - |\{y \in S: y = z\}| = & [k(2t+1) + t(n-k)] - (t+1)(n-k)\\
        = & k(2t+1) - (n-k) = k(2t + 2) - n \\
        \geq & 2n+2 - n > 0.
    \end{align*}
    The first inequality is by definition of $k$ and $t$.
\end{proof}

We also use a corollary of this for similar instances.
\begin{corollary} \label{cor:zeroopt}
    For any $n, k, t \in \mathbb{N}$ such that $0 < k \leq n$ and $t \geq n$, the optimal constant for an instance with $n$ agents and $2t+1$ points per agent and the same number of 0 entries as $S(n, k, t, z)$ is 0.
\end{corollary}

Since we know the optimal solution for this family of instances, we can use consistency to bound the constant returned for certain instances when the advice is equal to the optimal solution.

\begin{lemma}\label{lem:largeinst}
    For any $\gamma \in (0,2]$, fix $n, k \in \mathbb{N}$ such that $k = n \cdot \frac{\gamma}{\gamma+2}$ and $n > k+1$, and let $t = \lceil n(\gamma + 2)\rceil$. Fixing $D > 1$, for any strategyproof deterministic mechanism $\mathcal{M}$ that is $1 + \gamma$ consistent, given instance and correct advice $(S(n, k+1, t, D'), 0)$ for any $D' \geq D$, $\mathcal{M}$ must return some constant less than 1.
\end{lemma}


To prove this, we will provide a few helper results. 
\begin{lemma}\label{lem:underone}
    For any $\gamma \in (0,2]$, fix $n, k\in \mathbb{N}$ such that $k = n\cdot \frac{\gamma}{\gamma + 2}$ and $n > k + 1$ and let $t = \lceil n(\gamma + 2)\rceil$. Then for any deterministic mechanism $\mathcal{M}$ that is $1 + \gamma$ consistent, there exists $\delta = \delta(n, \gamma) > 0$ such that, given instance and correct advice $(S(n, k+1, t, 1),0)$, $\mathcal{M}$ must return some constant less than $1 - \delta$.
\end{lemma}
\begin{proof}
    Fix $\gamma \in (0,2]$ and $n,k,t$ that satisfy the conditions in the lemma statement. Observe by Lemma~\ref{lem:zeroopt} that the optimal constant for $S(n, k+1, t, 1)$ is $0$, as $k + 1 > 0$ and $t \geq n$ for any $\gamma$. Then to achieve $1 + \gamma$ consistency given instance $S$ and correct advice $0$, $\mathcal{M}$ must return some $c$ such that $\frac{R(c, S)}{R(0, S)} \leq 1+\gamma$. By Observation~\ref{obs:singlepeak}, we know that $R(c, S)$ increases as $c$ increases from $0$. Considering some $c \in [0,1]$, we can calculate the approximation ratio of selecting $c$ as the solution: 
    \begin{align*}
        \frac{R(c, S)}{R(0, S)} & = \frac{\frac{1}{nm}\sum_{y \in S}|y - c|}{\frac{1}{nm}\sum_{y \in S}|y - 0|} = \frac{c\cdot|\{y\in S: y = 0\}| + (1 - c)\cdot |y \in S: y = 1|}{|\{y \in S: y = 1\}|} \\
        & = \frac{c[(k+1)(2t+1) + t(n-k-1)] + (1-c)(n-k-1)(t+1)}{(n-k-1)(t+1)}\\
        & = 1 - \frac{c}{t+1}\cdot\frac{n}{n-k-1} + c \cdot \frac{2(k+1)}{n-k-1}.
    \end{align*}
    Now if we substitute $k$ this equals:
    \begin{align*} 
        & 1 - \frac{c}{t+1}\cdot \frac{n}{n-n\cdot\frac{\gamma}{\gamma+2}-1} + c \cdot \frac{2(n\cdot\frac{\gamma}{\gamma+2}+1)}{n-n\cdot\frac{\gamma}{\gamma+2}-1} \\
        = & 1 + c\left[\gamma + \frac{(2 + \gamma)^2}{2n-\gamma - 2} - \frac{n(\gamma + 2)}{(t+1)(2n-\gamma-2)}\right] \\
        \geq & 1 + c\left[\gamma + \frac{(2 + \gamma)^2}{2n - \gamma - 2} - \frac{n(\gamma + 2)}{n(\gamma + 2)(2n - \gamma - 2)}\right]\\
        \geq & 1 + c\left[\gamma + \frac{3}{2n - \gamma - 2}\right],
    \end{align*}
    with the first inequality by definition of $t$. Observe that for $\frac{R(c, S)}{R(0, S)}$ to be bounded above by $1+\gamma$, $c<1$ must hold. Thus, given $(S(n, k+1, t, 1), 0)$, there exists $\delta > 0$ such that $\mathcal{M}$ must return some constant $c < 1 - \delta$ to satisfy $1 + \gamma$ consistency.
\end{proof}
\begin{corollary}\label{cor:scale}
    For any $\gamma \in (0,2]$, fix $n, k\in \mathbb{N}$ such that $k = n\cdot \frac{\gamma}{\gamma + 2}$ and $n > k + 1$ (note such $n, k$always exist), and let  $t = \lceil n(\gamma + 2)\rceil$. Then for any deterministic mechanism $\mathcal{M}$ that is $1 + \gamma$ consistent, there exists $\delta = \delta(n, \gamma) > 0$ such that, for all $z > 0$, given instance and correct advice $(S(n, k+1, t, z), 0)$, $\mathcal{M}$ must return some constant less than $(1 - \delta)\cdot z$.
\end{corollary}
The proof is exactly as for Lemma~\ref{lem:underone}, scaling all risk values by $z$, and we also note that this means $\delta(n, \gamma)$ is also the same as in Lemma~\ref{lem:underone}.

\begin{proof}[Proof of Lemma~\ref{lem:largeinst}]
Take $\delta$ as in Lemma~\ref{lem:underone} and fix some $D > 1$. It suffices to show that for any $D'$ such that $1 + D' \cdot \frac{\delta}{2t+1} \geq D$, $\mathcal{M}$ must return some constant less than $1-\delta$ given instance $(S(n,k+1,t,1 + D'\cdot\frac{\delta}{2t+1}),0)$. For simplicity, let us denote $z(i) = 1+ i \cdot \frac{\delta}{2t+1}$ and $S(n, k+1, t, z(i))$ as $S^i$. Let us also define for any $j \leq n - k - 1$ the instance $S^i(j)$ as 
\[\{Z_t(0,0)^{k+1}\uplus Z_t(0,z(i+1))^j\uplus Z_t(0,z(i))^{n-k-j-1}\}.\]
Intuitively, as $j$ increases from 0 to $n-k-1$, we start from $S^i$ and incrementally change one agent type at a time to transition to $S^{i+1}$, and note that $S^i(j)$ has the same number of 0 entries as $S^i$, so by Corollary~\ref{cor:zeroopt} constant $0$ is the optimal solution to $S^i(j)$. 

We will first show that, for any $i < D'$, if $\mathcal{M}$ must return a constant less than $1 - \delta$ given instance $S^i$ and advice $0$, then given instance $S^{i+1}$ and advice 0, $\mathcal{M}$ must still return a constant less than $1 - \delta$. To do this, we will show by induction on $j$ that $\mathcal{M}$ must return a constant less than $1 - \delta$ given $(S^i(j),0)$ for all $j \leq n - k -1$. The result follows from the fact that $S^i(0) = S^i$ and $S^i(n - k - 1) = S^{i+1}$ 

\paragraph{Induction on $j$.} 
We want to show for any $j < n - k - 1$ that, if $\mathcal{M}$ must return a constant less than $1 - \delta$ given $(S^i(j),0)$, it must still return a constant less than $1 - \delta$ given $(S^i(j+1),0)$. Observe that $S^i(0) = S^i$, so the base case $j = 0$ is true from the claim statement. Assuming ,then, that $\mathcal{M}$ returns a constant less than $1 - \delta$ given $(S^i(j), 0)$ for some $j < n-k-1$, we want to show that given $S^i(j+1)$ mechanism $\mathcal{M}$ must still return some constant less than $1 - \delta$. First, we will show that given instance $(S^i(j+1), 0)$, $\mathcal{M}$ cannot return any $c \geq z(i+1)$ while maintaining $1 + \gamma$ consistency. Assume otherwise, then by Observation~\ref{obs:singlepeak} the global risk $R(c, S^i(j+1))$ increases as $c$ increases from $z(i+1)$ and $R(c, S^i(j+1)) \geq R({z(i+1)}, S^i(j+1))$. Then we calculate the approximation ratio $\frac{R({z(i+1)}, S^i(j+1))}{R(0, S^i(j+1))}$:
    \begin{align*}
        & \frac{|\{y \in S^i(j+1):y = 0\}|\cdot z(i+1) + |\{y \in S^i(j+1): y = z(i)\}|\cdot |z(i+1) - z(i)|}{|\{y \in S^i(j+1):y = z(i)\}|\cdot z(i) + |\{y \in S^i(j+1): y = z(i+1)\}|\cdot z(i+1)} \\
         = & \frac{|\{y \in S^i(j+1):y = 0\}|\cdot z(i+1) + (n - k - j - 2) \cdot |z(i+1) - z(i)|}{(n - k - j - 2)\cdot z(i) + (j + 1)\cdot z(i+1)} \\
         > & \frac{|\{y \in S^{i}(j+1):y = 0\}|\cdot z(i+1)}{(n - k - 1)\cdot z(i+1)} = \frac{|\{y \in S^{i+1}:y = 0\}|\cdot z(i+1)}{(n - k - 1)\cdot z(i+1)}\\
         = & \frac{R({z(i+1)}, S^{i+1})}{R(0, S^{i+1})} > \frac{R({(1-\delta)\cdot z(i+1)}, S^{i+1})}{R(0, S^{i+1})}.
    \end{align*}
    The first inequality is because $z(i+1) > z(i)$, and the second is by Observation~\ref{obs:singlepeak}. By Corollary~\ref{cor:scale}, we know that this is at least $1 + \gamma$, and, therefore, $\mathcal{M}$ returning $c \geq z(i+1)$ does not satisfy $1 + \gamma$ consistency.

    We can then show that given $(S^i(j+1), 0)$ mechanism $\mathcal{M}$ may not return any $c \in [1-\delta, z(i)]$. To see this, consider that instances $S^i(j)$ and $S^i(j+1)$ only differ by the report of one agent denoted $\bar{a}$, switching from type $S_{\bar{a}} = Z_t(0,z(i))$ to type $\bar{S}_{\bar{a}} = Z_t(0,z(i+1))$. Observe by Observation~\ref{obs:singlepeak} that $R_{\bar{a}}(c, S_{\bar{a}})$ increases as $c$ decreases from $z(i)$ since ${z(i)}$ is the personal risk minimizing function for $S_{\bar{a}}$. Since the induction hypothesis (on $j$) states that given $(S^i(j), 0)$, mechanism $\mathcal{M}$ returns a constant less than $1 - \delta$, returning $c \in [1 - \delta, z(i)]$ when the instance switches to $S^i(j+1)$ would decrease the risk for agent $\bar{a}$, and therefore if $\bar{a}$ had true type $S_{\bar{a}}$ they would have incentive to misreport their type as $\bar{S}_{\bar{a}}$, violating strategyproofness of $\mathcal{M}$.

    Now we consider the case when $\mathcal{M}$ returns some $c \in (z(i), z(i+1))$. Observe by Observation~\ref{obs:singlepeak} that $R_{\bar{a}}(c, S_{\bar{a}})$ increases as $c$ increases from $z(i)$, so $R_{\bar{a}}(c, S_{\bar{a}}) < R_{\bar{a}}({z(i+1)}, S_{\bar{a}})$. We can calculate the difference $R_{\bar{a}}({1-\delta}, S_{\bar{a}}) - R_{\bar{a}}({z(i+1)}, S_{\bar{a}})$ as follows:
    \begin{align*}
         & \frac{1}{2t+1}[|\{y \in S_{\bar{a}}: y = 0\}|(1 - \delta) + |\{y \in S_{\bar{a}}: y = z(i)\}|\cdot (z(i) - (1 - \delta)) \\
        & - |\{y \in S_{\bar{a}}: y = 0\}|\cdot z(i+1) - |\{y \in S_{\bar{a}}: y = z(i)\}|\cdot (z(i+1) - z(i))] \\
        = & \frac{1}{2t+1}[t \cdot (1 - \delta - z(i+1)) + (t+1)\cdot (z(i) - (1 - \delta) - (z(i+1) - z(i)))] \\
        = & \frac{1}{2t+1}[- t\cdot(\delta + (i+1)\frac{\delta}{2t+1}) + (t+1)(\delta + i\frac{\delta}{2t+1}-\frac{\delta}{2t+1})] \\
        = & \frac{1}{2t+1}\cdot \frac{\delta}{2t+1}\cdot i,
    \end{align*}
    and note that this is nonnegative for all $i \geq 0$. Then $R_{\bar{a}}(c, S_{\bar{a}}) < R_{\bar{a}}({z(i+1)}, S_{\bar{a}}) \leq R_{\bar{a}}({1 - \delta}, S_{\bar{a}})$, so if $\mathcal{M}$ returns $c \in (z(i), z(i+1))$ given $(S^i(j+1), 0)$, agent $\bar{a}$ with true type $S_{\bar{a}}$ would have incentive to misreport their type as $\bar{S}_a$ and decrease their personal risk. Then, $\mathcal{M}$ must return some $c < 1 - \delta$ given  $(S^i(j+1), 0)$. Setting $j = n - k - 2$, we get that $\mathcal{M}$ must return a constant below $1 - \delta$ given instance $S^i(n - k - 1) = S^{i+1}$, thus proving the claim.

We now know that, for any $i < D'$, if $\mathcal{M}$ must return a constant less than $1 - \delta$ given $(S^i, 0)$, then, given $(S^{i+1}, 0)$, mechanism $\mathcal{M}$ must still return a constant less that $1 - \delta$. We can now prove the main result that $\mathcal{M}$ must return some constant less than $1 - \delta$ given $(S^{D'},0)$, which we do by induction on $i$.

\paragraph{Induction on $i$.}
We will show that given $(S^i, 0)$ for any $i \leq D'$, $\mathcal{M}$ must return some constant less than $1 - \delta$. Observe that for $i = 0$, this holds by Lemma~\ref{lem:underone}. For any $i \leq D'$, we know by Lemma~\ref{lem:zeroopt} that $0$ is the optimal constant to $S^i$. We have shown above that if $\mathcal{M}$ must return a constant less than $1 - \delta$ given $(S^i,0)$, it must return a constant less than $1 - \delta$ given $(S^{i+1},0)$. Thus, the claim holds for all $i \leq D'$, and therefore $\mathcal{M}$ must return a constant less than 1 given instance $S^{D'} = S(n, k+1, t, D')$ and advice $0$, thus completing the proof.
\end{proof}

Finally, we want to bound the approximation of strategyproof mechanisms with a function that approaches as $n$ grows large $1 + 4/\gamma$. If we fix any (rational) $\gamma \in (0, 2]$ and $\beta < 1+4/\gamma$, define $r(n, D) := \frac{D-1}{D} \cdot \frac{4 + \gamma - \frac{9}{n}}{\gamma + \frac{4}{n}}$. We use the following two lemmas, the first of which allows us to consider $r(n, D)$ arbitrarily close to $1 + 4/\gamma$.

\begin{lemma}\label{lem:robustfunction}
For any rational $\gamma \in (0,2]$ and $\varepsilon > 0$, there exist $n, k \in \mathbb{N}$ such that $k = n\cdot\frac{\gamma}{\gamma+2}$, $n > k+1$ for which $r(n, D) > 1 + 4/\gamma - \varepsilon$ for all $D \geq \frac{2}{\varepsilon}+\frac{8}{\gamma \varepsilon}$. 
\end{lemma}
\begin{proof}

Observe that for any $n, D> 0$ we have that $r(n,D) < 1+4/\gamma$. 
    First, since $\lim_{n \to \infty}\frac{4 + \gamma -9/n}{\gamma + 4/n} = 1 + 4/\gamma$ and this is continuous in $n$, we can find $n'$ such that $\frac{4 + \gamma - 9/n'}{\gamma + 4/n'} \geq 1 + 4/\gamma - \varepsilon/2$. Then taking $D \geq \frac{2}{\varepsilon}(1 + 4/\gamma - \varepsilon/2)$, we can see that 
    \[r(n', D) = \left(1 - \frac{1}{D}\right)\cdot \frac{4 + \gamma - 9/n'}{\gamma + 4/n'} \geq \frac{1 + 4/\gamma - \varepsilon/2 - \varepsilon/2}{1 + 4/\gamma - \varepsilon/2} \cdot (1 + 4/\gamma - \varepsilon/2) = 1 + 4/\gamma - \varepsilon.\]
  
  Fix any $n', D$ that satisfy this. Then we can find $n, k \in \mathbb{N}$ such that $n \geq n'$, $k = n\cdot \frac{\gamma}{\gamma + 2}$, and $n > k+1$, and $r(n, D) > \beta$ must also hold because $r(n, D)$ is increasing in $n$. Concretely, first we let $\frac{\gamma}{\gamma + 2} = \frac{p}{q}$ for some $p, q \in \mathbb{N}$, which must exist because $\gamma$ is rational, and note that $q \geq p+1$. For any $\lambda \in \mathbb{N}$ such that $2\lambda q \geq n'$, we can set $n = 2\lambda q$ and $k = n\cdot \frac{\gamma}{\gamma+2} = 2\gamma p$, which satisfies $n \geq 2\gamma (p+1) > k + 1$.
\end{proof}
We can also bound the ratio achieved by certain outputs on instances of form $\{Z_t(D,0)^{k} \uplus\{Z_t(D,D)^{n-k}\}$.

\begin{restatable}{rLem}{lemBoundByFunc}\label{lem:boundbyfunc}
     For fixed number of agents $n > 0$ and any $a < 1$, $D > 1$, for all $D' \geq D$, if we define instance $S = \{Z_t(D',0)^{k} \uplus\{Z_t(D',D')^{n-k}\}$, then $\frac{R(a, S)}{R(D',S)} \geq r(n,D)$.
\end{restatable}

The proof of this we defer to Appendix~\ref{appendix:constantLB}. Note that any $D'$ that satisfies Lemma~\ref{lem:largeinst} also satisfies this result. 

We can now complete the proof of the main result. 

\begin{proof}[Proof of Theorem~\ref{thm:rationalconstantlb}]
    
Taking $D = \frac{2}{\varepsilon} + \frac{8}{\gamma \varepsilon}$, we know by Lemma~\ref{lem:robustfunction} that there exist $n, k \in \mathbb{N}$ such that $k = n\cdot \frac{\gamma}{\gamma + 2}, n > k+1$, and $r(n, D) > 1 + 4/\gamma - \varepsilon$. Observe that $T > D$ by definition of $T$. Letting $t = \lceil n(\gamma + 2)\rceil$, we know by Lemma~\ref{lem:largeinst} that $\mathcal{M}$ must return some constant less than 1 given instance and correct advice $(S(n,k+1,t,D'),0)$ for any $D' \geq D$, so we consider instance $S(n, k+1, t, T)$. Let us assign indices to the agents of $S^0 = S(n, k+1, t, T)$:
    \[S_1^0 = \ldots = S_{k+1}^0 = Z_t(0,0) \text{ and }S_{k+2}^0 = \ldots = S_n^0 = Z_t(0,T).\]

We now construct a series of instances for which we will prove that $\mathcal{M}$ must return a constant less than 1 in order to satisfy strategyproofness. We define agent types $S_j^i$, for $i \in [0, n]$ and $j\in[n]$:
\begin{align*}
i \leq k+1: S_j^i = \begin{cases}
    Z_t(T,0) & j \leq i \\
    Z_t(0,0) & i < j \leq k+1 \\
    Z_t(0,T) & k +1< j \leq n
    \end{cases},
i > k+1: S_j^i = \begin{cases}
        Z_t(T,0) & j \leq k+1 \\
        Z_t(T,T) & k+1 < j \leq i \\
        Z_t(0,T) & i < j \leq n
    \end{cases}.
\end{align*}
Let each instance $S^i = \biguplus_{l = 1}^nS_l^i$, which is consistent with the original definition of $S^0$, and observe that for every $i \in [n]$, $S^{i-1}$ and $S^i$ that only differ by the ith agent's type. If we let $S_{-j}^i = \biguplus_{l \in [n]\backslash j}S_l^i$, we can see that $S_{-i}^{i-1} = S_{-i}^i$. 

Our induction hypothesis is that, given $(S^i, 0)$, mechanism $\mathcal{M}$ must return a constant less than 1, which we know holds for $i = 0$. Now assume that this holds for $i - 1$ and $\mathcal{M}(S^{i-1}, 0) = a$ for some $a < 1$. We first consider when $1 \leq i \leq k + 1$. Recall that 
$S_i^{i-1} = Z_t(0,0)$ and $S_i^i = Z_t(T,0)$.
Both of these have unique personal risk minimizers $0$. Assume for contradiction that $\mathcal{M}(S^i, 0) = b$ for some $b \geq 1$. Then since $R_i(a, S_i^i) < R_i(b, S_i^i)$ by Observation~\ref{obs:singlepeak} for $0 < a < b$, using $S^{i-1} = S_i^{i-1} \uplus S_{-i}^i$ and $S^i = S_i^i \uplus S_{-i}^i$, we get that 
\[R_i(\mathcal{M}(S_i^{i-1} \uplus S_{-i}^i, 0), S_i^i) < R_i(\mathcal{M}(S_i^i \uplus S_{-i}^i, 0), S_i^i).\]
This means by fixing the types of all other agents to $S_{-i}^i$, an agent with type $S_i^i$ can misreport their type as $S_i^{i-1}$ and decrease their personal risk, which violates strategyproofness. Then $\mathcal{M}$ cannot return $b \geq 1$ given instance $(S^i, 0)$. Thus, the claim holds for $i = k + 1$. 

We can use a similar argument for all $k+1 < i \leq n$.
Recall that $S_i^{i-1} = Z_t(0,T)$ and $S_i^i = Z_t(T,T)$. Both of these have unique personal risk minimizer $T$. Assume for contradiction that $\mathcal{M}(S^i, 0) = b$ for some $b \geq 1$. Then, since $R_i(a, S_i^i) > R_i(b, S_i^i)$ by Observation~\ref{obs:singlepeak} and $T \geq b > a$, as $\mathcal{M}$ cannot output a constant above $T$,
we get that 
\[R_i(\mathcal{M}(S_i^{i-1} \uplus S_{-i}^i, 0), S_i^i) > R_i(\mathcal{M}(S_i^i \uplus S_{-i}^i, 0), S_i^i).\]
Then, if the types of all other agents is $S_{-i}^i$, an agent with type $S_{i-1}^i$ can misreport their type as $S_i^i$ and decrease their personal risk, again violating strategyproofness. Then, $\mathcal{M}$ cannot return $b$ having seen instance $(S^i, 0)$. 

Thus, $\mathcal{M}$ must return some $a < 1$ given $(S^n,0)$, with which we can now bound the robustness of $\mathcal{M}$. Recall that $S^n$ is defined by agent types
$S_1^n = \dots = S_{k+1}^n = Z_t(T,0)$ and $S_{k+2}^n = \dots = S_n^n = Z_t(T,T)$.
We know that $\mathcal{M}$ returns $a$ for some $a < 1$, so by Lemma~\ref{lem:boundbyfunc} we have that $\frac{R(a, S^n)}{R(T, S^n)} \geq r(n, D) > \beta$ and $\mathcal{M}$ cannot be $\beta$ robust. \qedhere
\end{proof}

\subsection{Homogeneous linear functions}\label{subsec:PFAlinear}

We leverage our results for constant functions to derive analogous results for the class $\mathcal{F}_{L}$  of homogeneous linear functions  over  $\mathbb{R}$, comprised of functions $f(x) = ax$ for some $a \in \mathbb{R}$. First, we provide a mechanism \lmech~for fitting homogeneous linear functions that is built on \mech. We then prove for \lmech \ similar results as for \mech:  it achieves the same consistency-robustness tradeoff,  a similar approximation with respect to advice error, and has analogous (group-)strategyproof properties. Finally, we provide a lower bound on the consistency-robustness tradeoff for  the family of homogeneous linear functions with upper bounded linear coefficient $a$.

\paragraph{The mechanism.} Let us first define for an instance-function pair the following mapping between the constant function setting and the homogeneous linear function setting. Let $x_{i,j} \in \mathbb{Z}$ for simplicity, but note that other cases can be dealt with by assigning weights to points. Given instance $S$, we take each pair $(x_{i,j}, y_{i,j})$ and make $|x_{i,j}|$ copies of the point $\frac{y_{i,j}}{x_{i,j}}$ to get $S_{C}$. Now consider a function $f \in \mathcal{F}_{L}$ such that $f(x) = ax$ for all $x$, which we denote $f_{L, a}$, with corresponding constant function $a$. Using the same parameter $\gamma \in [0,2]$ as in the \mech~mechanism, we design the following mechanism over homogeneous linear functions that leverages this mapping. We will let $\{\cdot\}^z$ for $z \in \mathbb{N}$ denotes $z$ copies of $(\cdot)$.

\vspace{.2cm}
\begin{algorithm}[H]
\setstretch{1.1}
\SetKwInOut{Input}{Input}
\Input{parameter $\gamma \in [0,2]$, dataset $S =\uplus_{i=1}^nS_i$, advice $\Tilde{f} = f_{L, \Tilde{a}} \in \mathcal{F}_L$}
$b \leftarrow \mech_{\gamma,\mathbb{R}}(\uplus_{i=1}^n\uplus_{j =1}^{|S_i|}\{\frac{y_{i,j}}{x_{i,j}}\}^{|x_{i,j}|}, \Tilde{a})$\\
\Return $f_{L, b}$ 
\caption{\lmechname (\lmech)}
\label{mech:lpfa}
\end{algorithm}
\vspace{.2cm}

We  denote the output of this as $\lmech_{\gamma}(S, \Tilde{f})$. The most computationally expensive part is when \mech~is run on $S_C$, resulting in total runtime $O(x_{\max}S\log|x_{\max}S|)$ for $x_{\max} = \max_{i,j}|x_{i,j}|$. Let us define the advice error in this setting for an instance $S$ with advice $\Tilde{f}(x) = f_{L, \Tilde{a}}$ and set of optimal solutions $F_{opt} = \arg\min_{f \in \mathcal{F}_L}R(f, S)$ as $\eta_L = \min_{f_{L, a^*} \in F_{opt}}\frac{|a^* - \Tilde{a}|}{R(f_{L, a^*}, S)}.$ Then we can show that \lmech~achieves the same guarantees as \mech~ by leveraging the fact that the risk in the homogeneous linear setting is a constant factor of the risk in the constant setting. 

\begin{theorem}\label{thm:LPFAguarantees}
    For any $\gamma \in (0, 2]$ and instance $S$,  $\lmech_\gamma$ is $1 + \gamma$ consistent and  $1+4/\gamma$ robust; it
  achieves approximation $\min\{1+4/\gamma, 1 + \gamma + \eta_L\cdot |S|/|S_C|\}$, and
if $\mech_{\gamma, \mathbb{R}}$ is (group)strategyproof on $S_C$, then $\lmech_\gamma$ is (group)strategyproof on $S.$ 
\end{theorem}

To prove this, we will use the notation $S_C$ as defined in \lmech~when it is clear which instance $S$ it is derived from. Recall that we define global empirical risk for any instance $S$ and function $f_{L,a}$ as
\[R(f_{L,a}, S) = \frac{1}{|S|}\sum_{(x_{i,j}, y_{i,j}) \in S}|f_{L,a}(x_{i,j}) - y_{i,j}| = \frac{1}{|S|}\sum_{(x_{i,j}, y_{i,j})\in S}|ax_{i,j} - y_{i,j}|.\]
We can show the following relationship between $R(f_{L,a}, S)$ and $R(f_a, S_C)$.

\begin{restatable}{rLem}{lemLinear}\label{lem:consttolin}
    Consider instance $S$ and function $f_{L,a} \in \mathcal{F}_{L}$. Then $R(f_a, S_C) = \frac{|S|}{|S_C|}R(f_{L,a}, S)$.
    Furthermore, for every $i \in [n]$, $R_i(f_a, S_{C,i}) = \frac{|S_i|}{|S_{C,i}|}R_i(f_{L,a}, S)$.
\end{restatable}
\begin{proof} The result follows immediately from this progression of calculations.
    \begin{align*}
        R(f_a, S_C) = & \frac{1}{|S_C|}\sum_{(x'_{i,j}, y'_{i,j}) \in S_C}|f_a(x'_{i,j}) - y'_{i,j}| \\
        = & \frac{1}{|S_C|}\sum_{i=1}^n\sum_{j=1}^{|S_i|}|x_{i,j}||a - \frac{y_{i,j}}{x_{i,j}}| \\
        = & \frac{1}{|S_C|}\sum_{i=1}^n\sum_{j=1}^{|S_i|}|ax_{i,j} - y_{i,j}| \\
        = & \frac{1}{|S_C|}\sum_{(x_{i,j}, y_{i,j}) \in S}|f_{L,a}(x_{i,j}) - y_{i,j}| = \frac{|S|}{|S_C|}R(f_{L,a}, S).
    \end{align*}
    The same argument can be made for the personal risk of any agent $i$.
\end{proof}

From this Lemma we can get the two following useful corollaries.

\begin{corollary}\label{cor:sameratio}
    For any instance $S$, for any $a, b \in \mathbb{R}$, the following holds:
    \[\frac{R(f_{L,a}, S)}{R(f_{L,b}, S)} = \frac{R(f_a, S_C)}{R(f_b, S_C)}.\]
\end{corollary}

\begin{corollary} \label{cor:sameminimizer}
    For any instance $S$, let $\mathcal{F}_{L,opt}$ be the set of risk minimizing homogeneous linear functions for $S$ and $\mathcal{F}_{C, opt}$ the set of risk minimizing constant functions for $S_C$. Then $f_{L,a} \in \mathcal{F}_{L, opt}$ if and only if $f_a \in \mathcal{F}_{C,opt}$.
\end{corollary}

We can now prove the consistency and robustness component of Theorem~\ref{thm:LPFAguarantees}.
\begin{lemma}\label{lem:consistrobustlin}
    For any $\gamma \in (0, 2]$, $\lmech_{\gamma}$ achieves $1 + \gamma$ consistency and $1+4/\gamma$ robustness.
\end{lemma}
\begin{proof}
    Consider an instance $S$ and advice $\Tilde{f} = f_{L,\Tilde{a}}\in \mathcal{F}_{L}$, and fix $\gamma \in (0, 2]$. Take any risk minimizing homogeneous linear function for $S$ denoted $f_{L,a^*}$. If we let $f_{L.b}$ be the output of~\lmech, then by applying Corollary~\ref{cor:sameratio} we get 
    \[\frac{R(\lmech_\gamma(S,f_{L,\Tilde{a}}),S)}{R(f_{L,a^*},S)} = \frac{R(f_{L,b},S)}{R(f_{L,a^*},S)} = \frac{R(f_b, S_C)}{R(f_{a^*},S_C)} = \frac{R(\mech_{\gamma,\mathbb{R}}(S_C,\Tilde{a}),S_C)}{R(f_{a^*},S_C)}.\]
   By Corollary~\ref{cor:sameminimizer}, we know that $f_{a^*}$ is a risk minimizing constant function for instance $S_C$, then since we know that \mech~ achieves robustness $1+4/\gamma$ from Theorem~\ref{thm:robustconsist}, this ratio is also bounded by $1+4/\gamma$ and \lmech~ also achieves $1+4/\gamma$ robustness.

    By applying the same argument to the case where the advice is correct, we get
     \[\frac{R(\lmech_{\gamma}(S, f_{L,\Tilde{a}}),S)}{R(f_{a^*}, S)} =  \frac{R(\mech_{\gamma, {\mathbb{R}}}(S_C, f_{a^*}),S_C)}{R(f_{a^*}, S_C)} \leq 1+\gamma.\]
     Thus $\lmech_{\gamma}$ also achieves the same consistency as $\mech_{\gamma, {\mathbb{R}}}$.
\end{proof}

We can also extend our approximation results for $\mech_{\gamma, \mathbb{R}}$ in terms of error to $\lmech_{\gamma}$. 

\begin{lemma}\label{lem:errorlin}
    For $\gamma \in (0, 2]$ and instance $S$, mechanism $\lmech_{\gamma}$ achieves approximation $\min\{ 1+4/\gamma, 1 + \gamma + \eta_l \frac{|S|}{|S_C|}\}$.
\end{lemma}
\begin{proof}
    Fix advice $f_{L,\Tilde{a}}$ and let $f_{L,b}$ be the output of \lmech~given $(S, f_{L,\Tilde{a}})$. We know by definition of $b$ in \lmech~ that $f_b = \mech_{\gamma,\mathbb{R}}(S_C,\Tilde{a})$. If we let $f_{L,a^*}=\arg\min_{f_{L,a^*} \in F_{opt}}\frac{|a^*-\Tilde{a}|}{R(f_{L,a^*},S)}$, we get that by applying Corollary~\ref{cor:sameratio} that
    \[\frac{R(f_{L,b}, S)}{R(f_{L,a^*}, S)} = \frac{R(b, S_C)}{R(f_{a^*}, S_C)}.\]
    Since $b$ is the output of $\mech_{\gamma, \mathbb{R}}(S_C,\Tilde{a})$ by definition of Mechanism~\ref{mech:lpfa}, and since by Corollary~\ref{cor:sameminimizer} we know that $f_{a^*}$ is an optimal constant function for $S_C$, we get that this equals
    \[\frac{R(\mech_{\gamma,\mathbb{R}}(S_C,\Tilde{a}),S_C)}{R(f_{a^*},S_C)} \leq 1+\gamma + \eta,\]
    the inequality coming from Theorem~\ref{thm:errortolerant}. Observing that $\eta_L = \eta\frac{|S|}{|S_C|}$ completes the proof.
\end{proof}

An interpretation of $\frac{|S|}{|S_C|}$ is the inverse of the average over all $|x_{i,j}|$. Thus, the impact of the error decreases as the input points get further from 0 on average. Finally, we can also apply our strategyproof results from $\mech_{\gamma,\mathbb{R}}$ to $\lmech_\gamma$.

\begin{lemma}\label{lem:strategyprooflin}
    For any $\gamma \in (0, 2]$ and instance $S$, if in the corresponding mapped set $S_C$ each $S_{C,i}$ has a unique personal risk minimizer over $\mathbb{R}$ for all $i \in [n]$, then $\lmech_{\gamma}$ is group-strategyproof on $S$.
\end{lemma}
This lemma follows from the fact that each agent $i$ can only control his mapped $S_{C,i}$, and since Lemma~\ref{lem:consttolin} shows that the personal risk any agent $i$ with type $S_i$ incurs from $\lmech$ is a constant factor of the personal risk an agent with type $S_{C,i}$ would incur from $\mech$, the existence of a coalition that can manipulate $\lmech$ to their benefit would imply the existence of a corresponding coalition in $\mech$, which violates the group-strategyproofness of $\mech$.

\begin{proof}[Proof of Theorem~\ref{thm:LPFAguarantees}.] Theorem~\ref{thm:LPFAguarantees} follows immediately from Lemma~\ref{lem:consistrobustlin}, Lemma~\ref{lem:errorlin}, and Lemma~\ref{lem:strategyprooflin}.
\end{proof}

\paragraph{The hardness result.}

Next, we show that for any $T$ sufficiently and arbitrarily large, no strategyproof mechanism for the class of homogeneous linear functions $\mathcal{F}_{L_{\leq T}} = \{f_a(x) = ax : a \leq T\}$ can achieve a better tradeoff than $1 + \gamma$ consistency and $1 + 4/\gamma$ robustness, the proof of which we defer to Appendix~\ref{appendix:linear}.

\begin{restatable}{rThm}{thmLinear}
\label{thm:linear}
Let $\epsilon > 0$. Then, for any rational $\gamma \in (0,2]$ and $T> \frac{2}{\varepsilon} + \frac{8}{\gamma \varepsilon}$, there is no deterministic strategyproof mechanism for the class of homogeneous linear functions $\mathcal{F}_{L_{\leq T}}$ that is $1+\gamma$ consistent and $1 + 4/ \gamma - \epsilon$ robust.
\end{restatable}

Thus, we have provided a lower bound for the consistency-robustness tradeoff of a slightly restricted class of homogeneous linear functions which matches the performance of \lmech.

\section{Strategyproof classification with advice} \label{sec:classification}

We now consider the binary classification problem where agents have a shared input, meaning each agent reports binary labels for the same set of points. The function class in this setting is a set of specific labelings of the shared points. \citet{meir2012algorithms} differentiate between selecting from two labelings and from more than two labelings. They show that deterministic mechanisms selecting from more than two labelings cannot achieve sublinear approximation, and randomized mechanisms cannot achieve better than $3 - 2/n$ approximation. We present results that show it is impossible to leverage correct advice to achieve consistency lower than these bounds while maintaining comparable robustness, in both the deterministic case and the randomized case (Section~\ref{subsec:genclassdet}). For two labelings, they provide a deterministic mechanism that achieves a 3 approximation and several randomized mechanisms that achieve a 2 approximation, both of which they show to be tight. Unlike the general decision setting, in the two labelings setting, we are able to design mechanisms that simultaneously achieve constant consistency and robustness (Section~\ref{subsec:binclass}). We are able to use \mech~ to get a deterministic mechanism with the same $1 + \gamma$ consistency and $1 + 4/\gamma$ robustness tradeoff, and we also provide our \binmech~ randomized mechanism that achieves, for $\gamma \in (0,1]$, $1+\gamma$ consistency and $1 + 1/\gamma$ robustness. All missing proofs are in Appendix~\ref{appendix:classification}.

\subsection{Any number of labelings}\label{subsec:genclassdet}

Let us first formally define the binary classification problem. The dataset $S$ is composed of a shared input $X = \{x_1, \dots, x_{|S_i|}\}$ with private labeling $Y_i\in \{0,1\}^{|S_i|}$ for each agent $i$. When it is clear what $X$ is, we may drop it and use $S$ to represent the $y$ values. The families of functions $\mathcal{C}$ that we now consider are comprised of labelings $c: X \to \{0, 1\}$ that assign specific labels to each $x_j$. The goal of the mechanism designer then is to select a labeling in $\mathcal{C}$. 

\paragraph{Deterministic mechanisms.} Designing deterministic classification mechanisms for arbitrary $\mathcal{C}$ is a very difficult problem. In fact, it was shown by~\citet{meir2012algorithms} that, even when restricting to selecting from sets of only three labelings, no deterministic strategyproof mechanism can achieve an approximation sublinear with respect to the number of agents. We prove that it is impossible to leverage correct advice to improve this approximation while maintaining bounded robustness.

\begin{restatable}{rThm}{thmSublinearConsist}\label{thm:sublinearconsist}
    There exists a function class $\mathcal{C}$ such that $|\mathcal{C}| = 3$ for which no deterministic strategyproof classification mechanism over $\mathcal{C}$ for $n > 10$ agents can simultaneously achieve bounded robustness and $o(n)$ consistency.
\end{restatable}

Let us first formally define the voting problem. For some finite set of alternatives (or candidates) $\mathcal{A}$, each of $n$ agents has a strict preference order $\succ_i$ over $\mathcal{A}$, and $\succ$ is the set of all agents' preference orders, also known as a preference profile. If we let $\mathcal{L}(\mathcal{A})$ be the set of all strict orders over $\mathcal{A}$, a voting rule is a function $v: \mathcal{L}(\mathcal{A})^n \to \mathcal{A}$. A voting rule is dictatorial if it always selects the favorite alternative of some agent $i'$. The following is a well-known result from the voting setting.

\begin{theorem}[\citet{gibbard1973manipulation,satterthwaite1975strategy}]\label{thm:gibbard}
    Let $v$ be a deterministic voting rule over alternatives $\mathcal{A}$ such that $|\mathcal{A}|\geq 3$. If $v$ is onto $\mathcal{A}$ and strategyproof, then it is dictatorial.
\end{theorem}

To prove Theorem~\ref{thm:sublinearconsist}, we will define a mapping from every deterministic strategyproof classification mechanism over $\mathcal{C}$ to a corresponding deterministic strategyproof voting mechanism over $\mathcal{C}$. Then, imposing this characterization of voting rules restricts the behavior of the classification mechanisms, allowing us to prove the lower bound.

\begin{table}\label{table}
\begin{center}
$\begin{array}{|c|c|c|c|c|}
    \hline
    \succ_i &  g_i(\succ_i) & Err(c_1, g_i(\succ_i)) & Err(c_2, g_i(\succ_i)) & Err(c_3, g_i(\succ_i)) \\
    \hline
    c_1 \succ_i c_2 \succ_i c_3 & \{1,1,1,1,1,1,1,1,1\} & 0 & 6 & 9 \\
    c_1 \succ_i c_3 \succ_i c_2 & \{1,1,1,1,1,1,0,0,0\} & 3 & 9 & 6 \\
    c_2 \succ_i c_1 \succ_i c_3 & \{0,0,0,0,1,1,1,1,1\} & 4 & 2 & 5 \\
    c_2 \succ_i c_3 \succ_i c_1 & \{0,0,0,0,0,0,1,1,1\} & 6 & 0 & 3 \\
    c_3 \succ_i c_1 \succ_i c_2 & \{0,0,1,1,1,1,0,0,0\} & 5 & 7 & 4 \\
    c_3 \succ_i c_2 \succ_i c_1 & \{0,0,0,0,0,0,0,0,0\} & 9 & 3 & 0 \\
    \hline
\end{array}$
\end{center}
\caption{Definition of function $g_i$ for each preference profile and relative number of errors of each mapped preference with respect to labelings in $\mathcal{C}$.}
\end{table}

\begin{proof}[Proof of Theorem~\ref{thm:sublinearconsist}]
    We present a subset of instances on which any strategyproof mechanism must behave a certain way, and by enforcing this behavior on a specially designed instance, we can lower bound consistency. Fixing some $n > 10$ and letting $m = 9$, we introduce the set of labelings we consider:\[c_1=\{1,1,1,1,1,1,1,1,1\}, c_2 = \{0,0,0,0,0,0,1,1,1\}, c_3 = \{0,0,0,0,0,0,0,0,0\},\]
and define $\mathcal{C} = \{c_1, c_2, c_3\}$.
    Let $\mathcal{M}$ be any deterministic strategyproof classification mechanism (with advice) over $\mathcal{C}$ that achieves bounded robustness. Assume advice $\Tilde{c} \in \mathcal{C}$ is given, and denote $\mathcal{M}_{\Tilde{c}}$ as the mechanism (without advice) induced by $\Tilde{c}$. In order for $\mathcal{M}$ to have bounded robustness, $\mathcal{M}_{\Tilde{c}}$ must have bounded approximation for any $\Tilde{c}$.

    The related voting setting is defined with $\mathcal{C}$ as the set of alternatives. In order to design a voting rule that uses $\mathcal{M}_{\Tilde{c}}$ as an intermediate step, we explicitly define a one-to-one mapping from preferences in $\mathcal{L}(\mathcal{C})$ to labelings in $\{0,1\}^m$ in Table~\ref{table}. Note that we use $Err(c, S_i)= \sum_{j = 1}^m \mathbf{1}(c(x_j) \neq y_{i,j})$ to denote the relative error of labeling $c$ with respect to $S_i$. We denote this $g_i$ for each $i \in [n]$ but note that they are all equivalent. Crucially, we ensure that the empirical personal risks an agent with labeling $g_i(\succ_i)$ experiences with respect to each labeling in $\mathcal{C}$ to agree with $\succ_i$, so for any $c, c' \in \mathcal{C}$ it must hold that 
    \begin{equation}\label{eq:maintainrisk}
        c \succ_i c' \to R_i(c, g_i(\succ_i)) < R_i(c', g_i(\succ_i)).
    \end{equation}  
    It is easy to verify using the relative number of errors in the table that this holds for each possible ordering $\succ_i$. Then we define the function $g$ that takes as input a preference profile $\succ \in \mathcal{L}(\mathcal{C})^n$ and returns a full set of $n$ labelings as follows: $g(\succ) = \{g_1(\succ_1), \dots, g_n(\succ_n)\}$.

Let $v = \mathcal{M}_{\Tilde{c}} \circ g$ and observe that $v$ is a voting rule. We prove that since $\mathcal{M}_{\Tilde{c}}$ has bounded approximation and is strategyproof, $v$ must be onto $\mathcal{C}$ and strategyproof, allowing us to apply Gibbard-Satterthwaite to $v$. To prove that $v$ is onto $\mathcal{C}$, it is sufficient to show that for every $c \in \mathcal{C}$, there must exist an instance $S$ such that $\mathcal{M}_{\Tilde{c}}(S) = c$ and $S = g(\succ)$ for some $\succ \in \mathcal{L}(\mathcal{C})^n$. Observe that for every $c$, there exists exactly one $\succ_i$ such that $R_i(c, g_i(\succ_i)) = 0$, and we denote this as $\succ_c$. For example, $\succ_{c_1}$ is defined by $c_1 \succ_{c_1} c_2 \succ_{c_1} c_3$. If all agents have labeling $S_i = g_i(\succ_c)$, then $c$ is the optimal solution, incurring 0 global risk. Thus returning $c'\neq c$ would result in unbounded approximation, so $\mathcal{M}_{\Tilde{c}}(S) = c$. 

We now prove that $v$ is strategyproof by contradiction. Assume instead that $v$ is not strategyproof, so there exists for some $i \in [n]$ preferences $\succ_{-i} \in \mathcal{L}(\mathcal{C})^{n-1}$ and  $\succ_i, \overline{\succ}_i \in \mathcal{L}(\mathcal{C})$ such that $v(\succ_{-i}, \overline{\succ}_i) \succ_i v(\succ_{-i}, \succ_i)$,
meaning agent $i$ with true preference $\succ_i$ gains by misreporting their preference as $\overline{\succ}_i$ when the other agents report $\succ_{-i}$. If we let $\bar{c} = v(\succ_{-i}, \overline{\succ}_i)$ and $c = v(\succ_{-i}, \succ_i)$, then by definition of $v$ we know that $\mathcal{M}_{\Tilde{c}}(g(\succ_{-i}, \overline{\succ}_i)) = \bar{c}$ and $\mathcal{M}_{\Tilde{c}}(g(\succ_{-i}, \succ_i)) = c$. By Equation~\ref{eq:maintainrisk}, we also know that $R_i(\bar{c}, g_i(\succ_i)) < R_i(c, g_i(\succ_i))$ Thus, when running $\mathcal{M}_{\Tilde{c}}$ on instance $g(\succ_{-i}, \succ_i)$, the ith agent can misreport their labeling as $g_i(\overline{\succ}_i)$ and lower their risk, violating the strategyproofness of $\mathcal{M}_{\Tilde{c}}$. Therefore $v$ is onto $\mathcal{C}$ and strategyproof, so by Theorem~\ref{thm:gibbard}, $v$ must be dictatorial.

Let $i'$ be the dictator in $v$, meaning on every voting problem instance, $v$ returns the top choice of $i'$. Then on any classification instance where for every $i$, $S_i = g_i(\succ_i)$ for some $\succ_i \in \mathcal{L}(\mathcal{C})$, $\mathcal{M}_{\Tilde{c}}$ must return $\arg\min_{c \in \mathcal{C}}R_i(c, S_{i'})$. Then consider the following instance defined by advice $\Tilde{c}$. For $i \neq i'$, let $S_i = g_i(\succ_{\Tilde{c}})$, and recall that this means $R_i(\Tilde{c}, S_i) = 0$. Then let $S_{i'} = g_{i'}(\succ_{c'})$ for some $c' \neq \Tilde{c}$. We know $\mathcal{M}_{\Tilde{c}}(S) = c'$, and we can calculate the global risk:
\[R(c', S) = \frac{1}{n}\sum_{i \neq i'}R_i(c', S_i) + R_{i'}(c', S_{i'}) = \frac{n-1}{n}R_i(c', g_i(\succ_{\Tilde{c}})) \geq \frac{n-1}{9n}.\]
Similarly, we calculate the global risk with respect to $\Tilde{c}$:
\[R(\Tilde{c}, S) = \frac{1}{n}\sum_{i \neq i'}R_i(\Tilde{c}, S_i) + R_{i'}(\Tilde{c}, S_{i'}) = \frac{1}{n}R_{i'}(\Tilde{c}, S_{i'}) \leq \frac{1}{n}.\]
Then since $n > 10$, $\Tilde{c}$ is the optimal solution and therefore $\frac{R(c', S)}{R(\Tilde{c}, S)} \geq \frac{n-1}{9} = O(n)$.
\end{proof}

\paragraph{Randomized mechanisms.} While it is impossible for a deterministic strategyproof mechanism selecting amongst three labelings to achieve sublinear approximation, \citet{meir2012algorithms} show that randomization in the mechanism makes it possible to achieve a 3 approximation while maintaining strategyproofness. The global risk incurred by a randomized mechanism $\mathcal{M}$ with output $f$ on instance $S$ is now defined $R(\mathcal{M}(S), S) = \mathbb{E}[R(f,S)]$, and the risk that agent $i$ with type $S_i$ seeks to minimize is $R_i(\mathcal{M}(S), S_i) = \mathbb{E}[R(f,S_i)]$, where the expectation is over the randomness in $\mathcal{M}$. It seems plausible, given these results, that having correct advice may allow for further improvement while maintaining constant approximation with incorrect advice. We show that, as with the deterministic case, this is impossible, and any improvement from 3 in the consistency results in robustness that scales with $n$.

We make a reasonable restriction that allows us to work with finite instances, namely that our mechanisms are $\mu$-granular for some $\mu > 0$. This means that if we represent any output of our mechanism as a probability vector $\mathbf{p}$ over the labelings, every $\mathbf{p}$ can be expressed as $\mu \cdot \mathbf{q}$ for some integer vector $\mathbf{q}$. Intuitively, every probability is a multiple of $\mu$. This is a reasonable assumption for any mechanism that can be implementable digitally. Then, our main theorem is as follows.

\begin{restatable}{rThm}{thmRandClass}\label{thm:randomclassification}
     Let $\mathcal{M}$ be a strategyproof randomized mechanism with advice, and let $\mathcal{M}$ be $\mu$-granular for some $\mu \in (0, 1]$. Then for any $\varepsilon > 0$ and $n \geq 2$, if $\mathcal{M}$  achieves consistency $3 - 2/n - \varepsilon$, then $\mathcal{M}$ is $\Omega(n)$ robust.
\end{restatable}

Similar to the proof of Theorem~\ref{thm:sublinearconsist}, this result leverages a characterization of strategyproof randomized mechanisms from voting theory, which, by restricting the behavior of such mechanisms on specific instances, results in our lower bound.


\subsection{Two labelings} \label{subsec:binclass}

We can reduce the problem of selecting from two labelings to selecting from the two constant labelings of all 0s ($c_0$) or all 1s ($c_1$), which we call the $\{c_0,c_1\}$-classification problem. This allows us to apply any binary classification mechanism to the $\{c_0,c_1\}$-classification problem.

\begin{restatable}{rThm}{thmToBinary}\label{thm:binto01}
        Let $\mathcal{M}$ be a strategyproof $\{c_0,c_1\}$-classification mechanism that is $\alpha$ consistent and $\beta$ robust. Then for any labelings $\{c', c''\}$ there exists a strategyproof classification mechanism $\mathcal{M}'$ over $\{c', c''\}$ that uses $\mathcal{M}$ as a blackbox that is $\alpha$ consistent and $\beta$ robust. Furthermore, $\mathcal{M}'$ is deterministic if $\mathcal{M}$ is deterministic. 
\end{restatable}

We can directly apply this theorem to the guarantees for $PFA_{\gamma,\{0,1\}}$ from Theorem~\ref{thm:robustconsist} to derive positive results in the two labelings setting as follows.
\begin{restatable}{rThm}{thmTwolabels}\label{thm:dettwolabels}

There exists a deterministic strategyproof binary decision mechanism over arbitrarily labelings $\{c',c''\}$ with parameter $\gamma \in (0, 2]$ that, using $\mech_{\gamma,\{0,1\}}$ as a blackbox, achieves $1 + \gamma$ consistency and $1+4/\gamma$ robustness. 
\end{restatable}

 We also provide a \binmechname mechanism for $\{c_0,c_1\}$-classification that is parametrized by $\gamma \in (0,1]$, denoted $\binmech_\gamma$, with the following guarantees.

\vspace{.2cm}
\begin{algorithm}[H]
\setstretch{1.1}
\SetKwInOut{Input}{Input}
\Input{ parameter $\gamma \in (0,1]$, dataset $S =\uplus_{i=1}^nS_i$, advice $\Tilde{c} \in \{c_0,c_1\}$}

$P \leftarrow \frac{1}{n}\sum_{i\in[n]} \mathbf{1}(R_i(c_1,S_i) \geq R_i(c_0, S_i))$, $N \leftarrow 1-P$\;
\If{$\Tilde{c} = c_1$}
{\Return $c_1$ w.p. $\frac{(P/\gamma)^2}{(P/\gamma)^2+(1-P)^2}$ and $c_0$ w.p. $\frac{(1-P)^2}{(P/\gamma)^2+(1-P)^2}$}
\Else{
 \Return $c_1$ w.p. $\frac{P^2}{P^2+((1-P)/\gamma)^2}$ and $c_0$ w.p. $\frac{((1-P)/\gamma)^2}{P^2+((1-P)/\gamma)^2}$}
\caption{\binmechname (\binmech)}
\label{mech:srda}
\end{algorithm}
\vspace{.2cm}

\begin{restatable}{rThm}{thmRandMech}\label{thm:randtwolabels}
    For any $\gamma \in (0,1]$, $\binmech_{\gamma}$ is strategyproof and achieves $1+\gamma$ consistency and $1+1/\gamma$ robustness. Furthermore, $\binmech_{\gamma}$ is group-strategyproof on any instance where $m$ is odd.
\end{restatable}

$\binmech_{\gamma}$ is clearly strategyproof because any agent misreporting their type can only increase the probability the mechanism selects a classifier that mislabels at least half of their points, thus weakly increasing their personal risk. By the same reasoning, it is group-strategyproof for odd $m$ as misreporting can only increase the probability the mechanism selects a classifier that mislabels \emph{more} than half of their points, increasing their personal risk. 

Let us define $N = 1-P$ as the proportion of agents that prefer $c_0$ over $c_1$. To prove the consistency and robustness results, we leverage the following lemma.

\begin{lemma}[\citet{meir2012algorithms}] \label{lem:binaryhelper}
    If $c^*=c_1$ is the optimal binary classifier for instance $S$, then if $r^*$ represents the minimum risk possible on instance $S$ and $P$ is defined as in Mechanism~\ref{mech:srda}, we get that
    \[1-r^* \leq \frac{1+P}{1-P}r^*.\]
    Similarly, if $c^* = c_0$ and $N$ is defined as in Mechanism~\ref{mech:srda}, we get that 
    \[1-r^* \leq \frac{1+N}{1-N} r^*.\]
\end{lemma}

With this, we can bound the consistency of \binmech.

\begin{lemma} \label{lem:binaryconsist}
    For any $\gamma \in(0,1]$, $\binmech_\gamma$ achieves $1+\gamma$ consistency.
\end{lemma}
\begin{proof}
Assume WLOG that $\Tilde{c} = c^* = c_1$. We can express the risk of $\binmech_\gamma$ as follows:
    \[R(\binmech_\gamma(S, \Tilde{c}),S) = \frac{(P/\gamma)^2}{(P/\gamma)^2+N^2}\cdot R(c_1, S) + \frac{N^2}{(P/\gamma)^2+N^2}\cdot R(c_0, S).\]
    Let $r^* = R(c_1, S)$ represent the minimum risk achievable for instance $S$ and observe that $R(c_0, S) = 1-r^*$. Substituting this, as well as $N = 1- P$, gives us the following:
    \begin{align*}R(\binmech_\gamma(S, \Tilde{c}), S) & = \frac{(P/\gamma)^2}{(P/\gamma)^2+(1-P)^2}\cdot r^* + \frac{(1-P)^2}{(P/\gamma)^2+(1-P)^2}\cdot(1-r^*) \\
    & \leq \frac{(P/\gamma)^2}{(P/\gamma)^2+(1-P)^2}\cdot r^* + \frac{(1-P)^2}{(P/\gamma)^2+(1-P)^2}\cdot \frac{1+P}{1-P}\cdot r^* \\
    & = \frac{1 + (1/\gamma^2-1)\cdot P^2}{1 - 2P + (1/\gamma^2+1)\cdot P^2}r^*.
    \end{align*}
    The inequality comes from applying Lemma~\ref{lem:binaryhelper}. Let us then define $h_\gamma(P) = \frac{1 + (1/\gamma^2-1)\cdot P^2}{1 - 2P + (1/\gamma^2+1)\cdot P^2}$. Taking the derivative we get that
    \[\frac{\partial h_\gamma}{\partial P}(P) = \frac{2-4P-2(\frac{1}{\gamma^2}-1)\cdot P^2}{(1-2P+(\frac{1}{\gamma^2}+1)\cdot P^2)^2}.\]
    Observe that for any $\gamma \in (0,1]$ the denominator is positive on the domain $P \in (0,1)$, and analyzing the quadratic formula in the numerator shows that $\frac{\partial h_\gamma}{\partial P}$ is equal to $0$ at $P = \frac{\gamma}{1+\gamma}$, negative for $P \in (0, \frac{\gamma}{1+\gamma})$, and positive for $P \in (\frac{\gamma}{1+\gamma}, 1)$. Thus, $h_\gamma(P)$ achieves its maximum (over the domain $P\in [0,1]$) at $P = \frac{\gamma}{1+\gamma}$. This allows us to achieve our consistency bound:
    \begin{align*}
        R(\binmech_\gamma(S,\Tilde{c}), S) & \leq h_\gamma\left(\frac{\gamma}{1+\gamma}\right)\cdot r^* \\
        & = \frac{1 + (\frac{1}{\gamma^2}-1)\cdot (\frac{\gamma}{1+\gamma})^2}{1 - 2\frac{\gamma}{1+\gamma} + (\frac{1}{\gamma^2}+1)\cdot (\frac{\gamma}{1+\gamma})^2} \cdot r^*  = (1+\gamma)\cdot r^*. \qedhere
    \end{align*}
    \end{proof}

Similarly, we can bound the robustness of \binmech.
\begin{lemma} \label{lem:binaryrobust}
    For any $\gamma \in (0,1]$, $\binmech_\gamma$ achieve $1+1/\gamma$ robustness.
\end{lemma}
\begin{proof}
    Assume WLOG that $\Tilde{c} = c_1$, then we can express the risk of $\binmech_\gamma$ as follows:
    \[R(\binmech_\gamma(S, \Tilde{c}),S) = \frac{(P/\gamma)^2}{(P/\gamma)^2+N^2}\cdot R(c_1, S) + \frac{N^2}{(P/\gamma)^2+N^2}\cdot R(c_0, S).\]
    By Lemma~\ref{lem:binaryconsist} we know that if $c^* = c_1$ this is bounded above by $(1+\gamma)\cdot R(c_1, S) = (1+\gamma)\cdot r^*$.
    Consider the other case where $c^* = c_0$, so $r^* = R(c_0, S)$ is the minimum risk achievable for instance $S$ and observe that $R(c_1, S) = 1-r^*$. Substituting this, as well as $P = 1-N$, we get that
    \begin{align*}
        R(\binmech_\gamma(S, \Tilde{c}), S) & = \frac{((1-N)/\gamma)^2}{((1-N)/\gamma)^2+N^2}\cdot (1-r^*) + \frac{N^2}{((1-N)/\gamma)^2+N^2}\cdot r^* \\
        & \leq \frac{((1-N)/\gamma)^2}{((1-N)/\gamma)^2+N^2}\cdot \frac{1+N}{1-N}\cdot r^* + \frac{N^2}{((1-N)/\gamma)^2+N^2}\cdot r^* \\
        & = \frac{(1-N)^2}{(1-N)^2 + \gamma^2N^2}\cdot \frac{1+N}{1-N}\cdot r^* + \frac{\gamma^2N^2}{(1-N)^2+\gamma^2N^2}\cdot r^* \\
        & = \frac{1+(\gamma^2-1)N^2}{1-2N+(\gamma^2+1)N^2}\cdot r^*.
    \end{align*}
     The inequality comes from applying Lemma~\ref{lem:binaryhelper}. Let us then define $h_\gamma(N) = \frac{1+(\gamma^2-1)N^2}{1-2N+(\gamma^2+1)N^2}$. Taking the derivative we get that
    \[\frac{\partial h_\gamma}{\partial N}(N) = \frac{2-4N-2(\gamma^2-1)\cdot N^2}{(1-2N+(\gamma^2+1)N^2)^2}.\]
    Observe that for any $\gamma \in (0,1]$ the denominator is positive on the domain $N \in (0,1)$, and analyzing the quadratic formula in the numerator shows that $\frac{\partial h_\gamma}{\partial N}$ is equal to $0$ at $N = \frac{1}{1+\gamma}$, negative for $N \in (0, \frac{1}{1+\gamma})$, and positive for $N \in (\frac{1}{1+\gamma}, 1)$. Thus, $h_\gamma(N)$ achieves its maximum (over the domain $N\in [0,1]$) at $N = \frac{1}{1+\gamma}$. This allows us to achieve our robustness bound:
     \begin{align*}
         R(\binmech_\gamma(S, \Tilde{c}), S) & \leq h_\gamma\left(\frac{1}{1+\gamma}\right)\cdot r^* \\
         & = \frac{1+(\gamma^2-1)(\frac{1}{1+\gamma})^2}{1-2\frac{1}{1+\gamma} + (\gamma^2+1)(\frac{1}{1+\gamma})^2}\cdot r^* = \left(1+1/\gamma\right)\cdot r^*.
     \end{align*}
    Observe that $\gamma \leq 1/\gamma$ for all $\gamma \in (0,1]$, so \binmech~ achieves $\max\{1+\gamma, 1+1/\gamma\} = 1+1/\gamma$ robustness.
\end{proof}
Combining Lemma~\ref{lem:binaryconsist} and Lemma~\ref{lem:binaryrobust} completes the proof of Theorem~\ref{thm:randtwolabels}. We also get the following results for classification over any two labelings $\{c',c''\}$ by applying Theorem~\ref{thm:binto01} to the guarantees in Theorem~\ref{thm:randtwolabels}.

\begin{theorem}
    There exists a strategyproof binary decision mechanism over arbitrary labelings $\{c',c''\}$ with parameter $\gamma \in (0,1]$ that, using $\binmech_\gamma$ as a blackbox, is $1+\gamma$ consistent and $1 + 1/\gamma$ robust.
\end{theorem}

\section{The learning-theoretic setting}\label{sec:learning}

We have developed mechanisms for when each agent has a finite set of input points over which they have preferences and incur personal risk, which we refer to as the \emph{decision-theoretic} setting. In this section, we extend our results to the \emph{learning-theoretic} setting where each agent $i$ has a (publicly known) distribution $\mathcal{D}_i$ over input space $\mathcal{X}$ representing the relative interest that agent has in each point, as well as a private function $y_i: \mathcal{X} \to \mathcal{Y}$. His personal risk over a given function is now defined
$R_i(f) = \mathbb{E}_{x \sim \mathcal{D}_i}[\ell(f(x), y_i(x))]$. Global risk is the risk of an agent chosen uniformly at random, denoted $R(f) = \frac{1}{n}\sum_{i = 1}^nR_i(f)$. 

Our goal is to design mechanisms to minimize $R(f)$ over a given family of functions $\mathcal{F}$. However, this problem is intractable for general distributions $\mathcal{D}_i$, so we consider the empirical risk minimization problem. For each agent $i$, we create dataset $S_i$ as a sampling of $|S_i|$ points $\{x_{i,1}, \dots, x_{i,|S_i|}\}$ drawn iid from $\mathcal{D}_i$ with agents' true labelings $y_i(x_{i,j})$, which together form agent types $S_i = \{(x_{i,j}, y_{i,j})\}_{j \in [|S_i|]}$. We now wish to use the empirical risk, as defined in the decision-theoretic setting, as a proxy for the risk with respect to the distributions, or \emph{statistical risk}.  

It is easy to see that no mechanism that only uses sampled data is strategyproof; for example, in binary classification, an agent may prefer label 0 for the vast majority of points in $\mathcal{X}$, but by chance, the sampled points $X$ may be ones he labels 1. We can, however, ensure that the amount an agent may gain through misreport is negligible by proving $\varepsilon$-strategyproofness and $\varepsilon$-group-strategyproofness for small $\varepsilon$, at which point we assume agents do not bother to misreport.

Strategyproof mechanisms that optimize with respect to empirical risk are likely to perform well (both approximation-wise and strategyproof-wise) in the learning setting if the difference between the two risks is close with high probability. We can show that this occurs when data set $S$ is composed of sufficiently large samples $S_i$ from each $\mathcal{D}_i$, applying PAC-learning techniques using Rademacher complexity\citep{bartlett2002rademacher} to get the following theorem.

\begin{restatable}{rThm}{thmLearning} \label{thm:learning}
Let $\varepsilon, \delta > 0$ and consider function class $\mathcal{F}$ with $O(\sqrt{m})$ Rademacher complexity with respect to $m$ samples from distributions $\mathcal{D}_i$. Consider samples $S_i$ with size $m = \Theta(\log(n/\delta)/\varepsilon^2)$ of iid points drawn from $\mathcal{D}_i$ for all $i \in [n]$ and dataset $S = \uplus_{i = 1}^nS_i$. Given mechanism $\mathcal{M}$ taking as input $S$ and an advice, with probability $1 - \delta$ the following holds:
    \begin{itemize}[leftmargin=*]
        \item if $\mathcal{M}$ is (group)strategyproof for all advice $\Tilde{f} \in \mathcal{F}$ with respect to agents with personal risk $R_i(f, S_i)$, then $\mathcal{M}$ is $\varepsilon$-(group)strategyproof with respect to agents with personal risk $R_i(f)$, and
        \item if $\mathcal{M}$ is $\alpha$ consistent and $\beta$ robust with respect to minimizing $R(f, S)$ and $f^* , f^*_S\in \mathcal{F}$ are optimal solutions with respect to $R(f)$ and $R(f, S)$ respectively, then $R(\mathcal{M}(S, f^*_S)) \leq \alpha \cdot R(f^*) + \frac{\alpha + 1}{2} \varepsilon$ and $R(\mathcal{M}(S, \Tilde{f})) \leq \beta \cdot R(f^*) + \frac{\beta + 1}{2} \varepsilon$ for any advice $\Tilde{f} \in \mathcal{F}$.
    \end{itemize}
\end{restatable}

We first introduce a theorem that, with large enough samples $S_i$, bounds the difference between the empirical and statistical risks.
\begin{theorem}[\citet{dekel2010incentive}] \label{thm:complexity}
    Let $\varepsilon, \delta>0$. If function class $\mathcal{F}$ has $O(\sqrt{m})$ Rademacher complexity with respect to $m$ samples from distribution $\mathcal{D}_i$, then for samples $S_i$ composed of $m_i = \Theta(\log(1/\delta)/\varepsilon^2)$ iid draws from $\mathcal{D}_i$, it holds that
    \begin{enumerate}
        \item $\PP(|R_i(f) - R_i(f, S_i)| \leq \varepsilon \text{ for all }f \in \mathcal{F}) \geq 1 - \delta$ and 
        \item $\PP(|R(f) - R(f, S)| \leq \varepsilon \text{ for all } f \in \mathcal{F}) \geq 1 - n\delta.$
    \end{enumerate}
\end{theorem}


\begin{proof}[Proof of Theorem~\ref{thm:learning}]

  By Theorem~\ref{thm:complexity}, we know that $|R_i(f) - R_i(f, S_i)| \leq \frac{\varepsilon}{2}$ for all $i \in [n]$ and $|R(f) - R(f, S)| \leq \frac{\varepsilon}{2}$ hold for all $f \in \mathcal{F}$ with probability at least $1 - \delta$. Assume that they do hold. Then the first statement comes directly from Theorem 6.1~\citep{dekel2010incentive} by fixing the advice and considering the subsequent without advice mechanism.

    For the consistency result in the second statement, given correct advice (with respect to empirical risk) $\Tilde{f} = f^*_S$ let $M(S, \Tilde{f}) = f$. Then,
    \[R(f) \leq R(f, S) + \frac{\varepsilon}{2} \leq \alpha\cdot R(f^*_S, S) + \frac{\varepsilon}{2} \leq \alpha \cdot R(f^*, S) + \frac{\varepsilon}{2}, \]
    where the inequalities come from applying the global risk assumption, consistency of $M$, and the fact that $f^*_S$ minimizes $R(f, S)$. We finally apply that $|R(f^*) - R(f^*, S)| \leq \frac{\varepsilon}{2}$ to complete the proof.

    A similar argument can be applied for the robustness result. 
    \end{proof}

We can apply this to our mechanisms, which we do below for \mech~by applying Theorem~\ref{thm:robustconsist}. 

\begin{theorem}
    Consider fitting a constant function over input space $\mathcal{X} = \mathbb{R}^k$ and output space $\mathcal{Y} = \arbset\subseteq \mathbb{R}$. Let $\varepsilon, \delta > 0$ and consider samples $S_i$ with odd $m = \Theta(\log(n/\delta)\varepsilon^2)$ drawn from distributions $\mathcal{D}_i$ with dataset $S = \uplus_{i = 1}^nS_i$. Then with probability $1 - \delta$, the following holds:
    \begin{itemize}
        \item for all $\gamma \in [0,2]$ and advice $\Tilde{a} \in {\arbset}$, $\mech_{\gamma, {\arbset}}(S, \Tilde{a})$ is $\varepsilon$-group-strategyproof with respect to agents with personal risk $R_i(f)$, and
        \item for all $\gamma \in (0, 2]$ and given $a^*_S \in \arg\min_{a \in {\arbset}}R(a, S),$ $a^* \in \arg\min_{a \in {\arbset}}R(a)$, $\mech_{\gamma, {\arbset}}(S, a^*_S)$ achieves approximation $(1 + \gamma)\cdot\left(1 + \Omega(1/n) / R(a^*)\right)$ and for all $\Tilde{a} \in {\arbset}$, $\mech_{\gamma, {\arbset}}(S, \Tilde{a})$ achieves approximation $\left(1 + 4/\gamma\right)\cdot \left(1 + \Omega(1/n)/R(a^*)\right)$. 
    \end{itemize}
\end{theorem}
This result encapsulates selecting over all constant values in $\mathbb{R}$ and $\{0,1\}$ for the $\{c_0,c_1\}$-classification problem. Observe that the odd size $m$ condition is  reasonable since $S_i$ are now samples from a distribution. We obtain similar extension for \lmech.

\begin{theorem}
    Consider fitting a homogeneous linear function over input and output spaces $\mathcal{X} = \mathcal{Y} = \mathbb{R}$. Let $\varepsilon, \delta > 0$ and consider samples $S_i$ with size $m = \Theta(\log(n/\delta)\varepsilon^2)$ drawn from distributions $\mathcal{D}_i$ with dataset $S = \uplus_{i = 1}^nS_i$. Then with probability $1 - \delta$ the following holds:
    \begin{itemize}
        \item for all $\gamma \in [0,2]$ and advice $\Tilde{f} \in \mathcal{F}_{L}$, $\lmech_{\gamma}(S, \Tilde{f})$ is $\varepsilon$-strategyproof with respect to agents with personal risk $R_i(f)$ if $\lmech_{\gamma}(S,\Tilde{f})$ is strategyproof with respect to agents' with personal risk $R_i(f, S_i)$, and
        \item for all $\gamma \in (0, 2]$ and given $f^*_S \in \arg\min_{f \in \mathcal{F}_{L}}R(f, S)$ and $f^* \in \arg\min_{f \in \mathcal{F}_{L}}R(f)$, $\lmech_{\gamma}(S, f^*_S)$ achieves approximation $(1 + \gamma)\cdot\left(1 + \Omega(1/n)/ R(f^*)\right)$ and for all $\Tilde{f} \in \mathcal{F}_{L}$, $\lmech_{\gamma}(S, \Tilde{f})$ achieves approximation $\left(1+4/\gamma\right)\cdot \left(1 + \Omega(1/n)/R(f^*)\right)$. 
    \end{itemize}
\end{theorem}
This follows directly from applying Theorem~\ref{thm:linear} regarding the performance of \lmech~to Theorem~\ref{thm:learning}. Note that a sufficient condition for $\mech_\gamma(S,\Tilde{f})$ to be strategyproof is if each $S_{C,i}$ in $S_C$ as defined in the mechanism has odd size. Finally, we extend the guarantees of \binmech. 
\begin{theorem}
    Consider selecting from the two constant binary functions $c_0, c_1$ over input space $\mathcal{X} = \mathbb{R}^k$ and output space $\mathcal{Y} = \{0, 1\}$. Let $\varepsilon, \delta > 0$ and consider samples $S_i$ with size odd $m = \Theta(\log(n/\delta)\varepsilon^2)$ drawn from distributions $\mathcal{D}_i$ with dataset $S = \uplus_{i = 1}^nS_i$. Then for all $\gamma \in (0,1]$, with probability $1 - \delta$ the following holds:
    \begin{itemize}
        \item for any advice $\Tilde{c} \in \{c_0,c_1\}$, $\binmech_{\gamma}(S,\Tilde{c})$ is $\varepsilon$-group-strategyproof with respect to agents with personal risk $R_i(f)$, and
        \item given $c^*_S \in \arg\min_{c \in \{c_0,c_1\}}R(c,S)$ and $c^* \in \arg\min_{c \in \{c_0,c_1\}}R(c)$, $\binmech_\gamma(S,c^*_S)$ achieves approximation $(1+\gamma)\cdot(1 + \Omega(1/n)/R(c^*))$ and for any $\Tilde{c} \in \{c_0,c_1\}$, $\binmech_\gamma(S,\Tilde{c})$ achieves approximation $(1 + 1/\gamma)\cdot(1 + \Omega(1/n)/R(c^*))$.
    \end{itemize}
\end{theorem}
This result follows directly from applying Theorem~\ref{thm:randtwolabels} to Theorem~\ref{thm:learning}, recognizing that binary loss is equivalent to absolute loss when using 0 and 1 as the labels. We only consider selecting between the two constant binary functions because they are well-defined over all of $\mathcal{X}$, whereas the general labelings we consider in the empirical binary classification setting are not.

\section*{Acknowledgments}
    The authors were supported by NSF grants CCF-2210501 and IIS-2147361.

\newpage

\bibliographystyle{abbrvnat}
\bibliography{biblio}

\begin{thebibliography}{38}
\providecommand{\natexlab}[1]{#1}
\providecommand{\url}[1]{\texttt{#1}}
\expandafter\ifx\csname urlstyle\endcsname\relax
  \providecommand{\doi}[1]{doi: #1}\else
  \providecommand{\doi}{doi: \begingroup \urlstyle{rm}\Url}\fi

\bibitem[Agrawal et~al.(2022)Agrawal, Balkanski, Gkatzelis, Ou, and Tan]{agrawal2022learning}
P.~Agrawal, E.~Balkanski, V.~Gkatzelis, T.~Ou, and X.~Tan.
\newblock Learning-augmented mechanism design: Leveraging predictions for facility location.
\newblock In \emph{Proceedings of the 23rd ACM Conference on Economics and Computation}, pages 497--528, 2022.

\bibitem[Ahmadi et~al.(2021)Ahmadi, Beyhaghi, Blum, and Naggita]{ahmadi2021strategic}
S.~Ahmadi, H.~Beyhaghi, A.~Blum, and K.~Naggita.
\newblock The strategic perceptron.
\newblock In \emph{Proceedings of the 22nd ACM Conference on Economics and Computation}, pages 6--25, 2021.

\bibitem[Balkanski et~al.(2023)Balkanski, Gkatzelis, and Tan]{balkanski2022strategyproof}
E.~Balkanski, V.~Gkatzelis, and X.~Tan.
\newblock Strategyproof scheduling with predictions.
\newblock In \emph{14th Innovations in Theoretical Computer Science Conference}, 2023.

\bibitem[Balkanski et~al.(2024{\natexlab{a}})Balkanski, Gkatzelis, and Shahkarami]{balkanski2024randomized}
E.~Balkanski, V.~Gkatzelis, and G.~Shahkarami.
\newblock Randomized strategic facility location with predictions.
\newblock \emph{arXiv preprint arXiv:2409.07142}, 2024{\natexlab{a}}.

\bibitem[Balkanski et~al.(2024{\natexlab{b}})Balkanski, Gkatzelis, Tan, and Zhu]{balkanski2023online}
E.~Balkanski, V.~Gkatzelis, X.~Tan, and C.~Zhu.
\newblock Online mechanism design with predictions.
\newblock In \emph{25th {ACM} Conference on Economics and Computation}, 2024{\natexlab{b}}.

\bibitem[Barak et~al.(2024)Barak, Gupta, and Talgam-Cohen]{barak2024mac}
Z.~Barak, A.~Gupta, and I.~Talgam-Cohen.
\newblock Mac advice for facility location mechanism design.
\newblock \emph{arXiv preprint arXiv:2403.12181}, 2024.

\bibitem[Bartlett and Mendelson(2002)]{bartlett2002rademacher}
P.~L. Bartlett and S.~Mendelson.
\newblock Rademacher and gaussian complexities: Risk bounds and structural results.
\newblock \emph{Journal of Machine Learning Research}, 3\penalty0 (Nov):\penalty0 463--482, 2002.

\bibitem[Berger et~al.(2024)Berger, Feldman, Gkatzelis, and Tan]{berger2023optimal}
B.~Berger, M.~Feldman, V.~Gkatzelis, and X.~Tan.
\newblock Optimal metric distortion with predictions.
\newblock In \emph{25th {ACM} Conference on Economics and Computation}, 2024.

\bibitem[Caragiannis and Kalantzis(2024)]{caragiannis2024randomized}
I.~Caragiannis and G.~Kalantzis.
\newblock Randomized learning-augmented auctions with revenue guarantees.
\newblock In \emph{Proceedings of the Thirty-Third International Joint Conference on Artificial Intelligence}, 2024.

\bibitem[Caro and Gallien(2010)]{caro2010inventory}
F.~Caro and J.~Gallien.
\newblock Inventory management of a fast-fashion retail network.
\newblock \emph{Operations research}, 58\penalty0 (2):\penalty0 257--273, 2010.

\bibitem[Caro et~al.(2010)Caro, Gallien, D{\'\i}az, Garc{\'\i}a, Corredoira, Montes, Ramos, and Correa]{caro2010zara}
F.~Caro, J.~Gallien, M.~D{\'\i}az, J.~Garc{\'\i}a, J.~M. Corredoira, M.~Montes, J.~A. Ramos, and J.~Correa.
\newblock Zara uses operations research to reengineer its global distribution process.
\newblock \emph{Interfaces}, 40\penalty0 (1):\penalty0 71--84, 2010.

\bibitem[Chen et~al.(2018)Chen, Podimata, Procaccia, and Shah]{chen2018strategyproof}
Y.~Chen, C.~Podimata, A.~D. Procaccia, and N.~Shah.
\newblock Strategyproof linear regression in high dimensions.
\newblock In \emph{Proceedings of the 2018 ACM Conference on Economics and Computation}, pages 9--26, 2018.

\bibitem[Colini-Baldeschi et~al.(2024)Colini-Baldeschi, Klumper, Sch{\"a}fer, and Tsikiridis]{colini2024trust}
R.~Colini-Baldeschi, S.~Klumper, G.~Sch{\"a}fer, and A.~Tsikiridis.
\newblock To trust or not to trust: Assignment mechanisms with predictions in the private graph model.
\newblock \emph{arXiv preprint arXiv:2403.03725}, 2024.

\bibitem[Cummings et~al.(2015)Cummings, Ioannidis, and Ligett]{cummings2015truthful}
R.~Cummings, S.~Ioannidis, and K.~Ligett.
\newblock Truthful linear regression.
\newblock In \emph{Conference on Learning Theory}, pages 448--483. PMLR, 2015.

\bibitem[Dekel et~al.(2010)Dekel, Fischer, and Procaccia]{dekel2010incentive}
O.~Dekel, F.~Fischer, and A.~D. Procaccia.
\newblock Incentive compatible regression learning.
\newblock \emph{Journal of Computer and System Sciences}, 76\penalty0 (8):\penalty0 759--777, 2010.

\bibitem[Dong et~al.(2018)Dong, Roth, Schutzman, Waggoner, and Wu]{dong2018strategic}
J.~Dong, A.~Roth, Z.~Schutzman, B.~Waggoner, and Z.~S. Wu.
\newblock Strategic classification from revealed preferences.
\newblock In \emph{Proceedings of the 2018 ACM Conference on Economics and Computation}, pages 55--70, 2018.

\bibitem[El-Mhamdi et~al.(2023)El-Mhamdi, Farhadkhani, Guerraoui, and Hoang]{el2023strategyproofness}
E.-M. El-Mhamdi, S.~Farhadkhani, R.~Guerraoui, and L.-N. Hoang.
\newblock On the strategyproofness of the geometric median.
\newblock In \emph{International Conference on Artificial Intelligence and Statistics}, pages 2603--2640. PMLR, 2023.

\bibitem[Ergun et~al.(2022)Ergun, Feng, Silwal, Woodruff, and Zhou]{ergun2021learning}
J.~C. Ergun, Z.~Feng, S.~Silwal, D.~P. Woodruff, and S.~Zhou.
\newblock Learning-augmented $ k $-means clustering.
\newblock \emph{The Tenth International Conference on Learning Representations}, 2022.

\bibitem[Gamlath et~al.(2022)Gamlath, Lattanzi, Norouzi-Fard, and Svensson]{gamlath2022approximate}
B.~Gamlath, S.~Lattanzi, A.~Norouzi-Fard, and O.~Svensson.
\newblock Approximate cluster recovery from noisy labels.
\newblock In \emph{Conference on Learning Theory}, pages 1463--1509. PMLR, 2022.

\bibitem[Ghalme et~al.(2021)Ghalme, Nair, Eilat, Talgam-Cohen, and Rosenfeld]{ghalme2021strategic}
G.~Ghalme, V.~Nair, I.~Eilat, I.~Talgam-Cohen, and N.~Rosenfeld.
\newblock Strategic classification in the dark.
\newblock In \emph{International Conference on Machine Learning}, pages 3672--3681. PMLR, 2021.

\bibitem[Gibbard(1973)]{gibbard1973manipulation}
A.~Gibbard.
\newblock Manipulation of voting schemes: a general result.
\newblock \emph{Econometrica: journal of the Econometric Society}, pages 587--601, 1973.

\bibitem[Gkatzelis et~al.(2022)Gkatzelis, Kollias, Sgouritsa, and Tan]{gkatzelis2022improved}
V.~Gkatzelis, K.~Kollias, A.~Sgouritsa, and X.~Tan.
\newblock Improved price of anarchy via predictions.
\newblock In \emph{Proceedings of the 23rd ACM Conference on Economics and Computation}, pages 529--557, 2022.

\bibitem[Hardt et~al.(2016)Hardt, Megiddo, Papadimitriou, and Wootters]{hardt2016strategic}
M.~Hardt, N.~Megiddo, C.~Papadimitriou, and M.~Wootters.
\newblock Strategic classification.
\newblock In \emph{Proceedings of the 2016 ACM Conference on Innovations in Theoretical Computer Science}, pages 111--122, 2016.

\bibitem[Istrate and Bonchis(2022)]{istrate2022mechanism}
G.~Istrate and C.~Bonchis.
\newblock Mechanism design with predictions for obnoxious facility location.
\newblock \emph{arXiv preprint arXiv:2212.09521}, 2022.

\bibitem[Lu et~al.(2024)Lu, Wan, and Zhang]{lu2023competitive}
P.~Lu, Z.~Wan, and J.~Zhang.
\newblock Competitive auctions with imperfect predictions.
\newblock In \emph{25th {ACM} Conference on Economics and Computation}, 2024.

\bibitem[Lykouris and Vassilvitskii(2021)]{lykouris2021competitive}
T.~Lykouris and S.~Vassilvitskii.
\newblock Competitive caching with machine learned advice.
\newblock \emph{Journal of the ACM (JACM)}, 68\penalty0 (4):\penalty0 1--25, 2021.

\bibitem[Meir and Rosenschein(2011)]{meir2011strategyproof}
R.~Meir and J.~S. Rosenschein.
\newblock Strategyproof classification.
\newblock \emph{ACM SIGecom Exchanges}, 10\penalty0 (3):\penalty0 21--25, 2011.

\bibitem[Meir et~al.(2008)Meir, Procaccia, and Rosenschein]{meir2008strategyproof}
R.~Meir, A.~D. Procaccia, and J.~S. Rosenschein.
\newblock Strategyproof classification under constant hypotheses: A tale of two functions.
\newblock In \emph{AAAI}, volume~8, pages 126--131, 2008.

\bibitem[Meir et~al.(2010)Meir, Procaccia, and Rosenschein]{meir2010limits}
R.~Meir, A.~D. Procaccia, and J.~S. Rosenschein.
\newblock On the limits of dictatorial classification.
\newblock In \emph{Proceedings of the 9th International Conference on Autonomous Agents and Multiagent Systems: volume 1-Volume 1}, pages 609--616. Citeseer, 2010.

\bibitem[Meir et~al.(2011)Meir, Almagor, Michaelly, and Rosenschein]{meir2011tight}
R.~Meir, S.~Almagor, A.~Michaelly, and J.~S. Rosenschein.
\newblock Tight bounds for strategyproof classification.
\newblock In \emph{Proceedings of the 10th International Conference on Autonomous Agents and Multiagent Systems}, 2011.

\bibitem[Meir et~al.(2012)Meir, Procaccia, and Rosenschein]{meir2012algorithms}
R.~Meir, A.~D. Procaccia, and J.~S. Rosenschein.
\newblock Algorithms for strategyproof classification.
\newblock \emph{Artificial Intelligence}, 186:\penalty0 123--156, 2012.

\bibitem[Nguyen et~al.(2023)Nguyen, Chaturvedi, and Nguyen]{nguyen2022improved}
T.~Nguyen, A.~Chaturvedi, and H.~L. Nguyen.
\newblock Improved learning-augmented algorithms for k-means and k-medians clustering.
\newblock \emph{The Eleventh International Conference on Learning Representations}, 2023.

\bibitem[Perote and Perote-Pena(2004)]{perote2004strategy}
J.~Perote and J.~Perote-Pena.
\newblock Strategy-proof estimators for simple regression.
\newblock \emph{Mathematical Social Sciences}, 47\penalty0 (2):\penalty0 153--176, 2004.

\bibitem[Prasad et~al.(2024)Prasad, Balcan, and Sandholm]{prasad2024bicriteria}
S.~Prasad, M.-F.~F. Balcan, and T.~Sandholm.
\newblock Bicriteria multidimensional mechanism design with side information.
\newblock \emph{Advances in Neural Information Processing Systems}, 36, 2024.

\bibitem[Procaccia and Tennenholtz(2013)]{procaccia2013approximate}
A.~D. Procaccia and M.~Tennenholtz.
\newblock Approximate mechanism design without money.
\newblock \emph{ACM Transactions on Economics and Computation}, 1\penalty0 (4):\penalty0 1--26, 2013.

\bibitem[Raman and Tewari(2024)]{raman2024online}
V.~Raman and A.~Tewari.
\newblock Online classification with predictions.
\newblock \emph{Advances in Neural Information Processing Systems}, 36, 2024.

\bibitem[Satterthwaite(1975)]{satterthwaite1975strategy}
M.~A. Satterthwaite.
\newblock Strategy-proofness and arrow's conditions: Existence and correspondence theorems for voting procedures and social welfare functions.
\newblock \emph{Journal of economic theory}, 10\penalty0 (2):\penalty0 187--217, 1975.

\bibitem[Xu and Lu(2022)]{xu2022mechanism}
C.~Xu and P.~Lu.
\newblock Mechanism design with predictions.
\newblock In \emph{Proceedings of the Thirty-First International Joint Conference on Artificial Intelligence}, pages 571--577, 2022.

\end{thebibliography}

\newpage

\appendix
\section{Missing proofs from Section~\ref{sec:regression}} \label{appendix:regression}

\subsection{Missing proofs from Section~\ref{subsec:PFAneg} (lower bound on consistency-robustness tradeoff)}\label{appendix:constantLB}

\lemBoundByFunc*
\begin{proof}
    Recall that $S^n = \{Z_t(D',0)^{k+1}\uplus Z_t(D',D')^{n-k-1}\}$. Then we can perform the following series of computations.
\begin{align*}
\frac{R(a, S^n)}{R(D',S^n)} & \geq \frac{R(1, S^n)}{R(D', S^n)} = \frac{\frac{1}{n(2t+1)}\sum_{y \in S^n}|y - 1|}{\frac{1}{n(2t+1)}\sum_{y \in S^n}|y - D'|}\\
& = \frac{|\{y\in S^n: y = 0\}| + (D'-1)\cdot|\{y \in S^n: y = D'\}|}{D'\cdot |\{y \in S^n: y = 0\}|}\\
& = \frac{1}{D'} + \frac{D'-1}{D'}\cdot \frac{t(k+1) + (2t+1)(n-k-1)}{(t+1)(k+1)}\\
& \geq \frac{D'-1}{D'}\cdot \frac{t(k+1) + (2t+1)(n-k-1)}{(t+1)(k+1)}\\
& \geq \frac{D-1}{D}\cdot \frac{t(k+1) + (2t+1)(n-k-1)}{(t+1)(k+1)}.
\end{align*}
The last line comes from the fact that $D' \geq D$. We can then decompose the second term to get
\begin{align*}
    \frac{t(k+1) + (2t+1)(n-k-1)}{(t+1)(k+1)} & = 1 - \frac{1}{t+1} + \frac{2(n - k - 1)}{k+1} - \frac{n - k - 1}{(t+1)(k+1)} \\
    & \geq \frac{k + 2(n - k - 1)}{k+1} - \frac{n}{t(k+1)} \\
    & = \frac{k + 2(n - k - 1) - \frac{n}{t}}{k+1}.
\end{align*}
Substituting in $k = n\cdot\frac{\gamma}{\gamma + 2}$ and $t = \lceil n(\gamma + 2)\rceil$, we get that this equals
\begin{align*}
    \frac{n\cdot\frac{\gamma}{\gamma + 2} + 2(n - n\cdot\frac{\gamma}{\gamma + 2} - 1) - \frac{n}{\lceil n(\gamma + 2)\rceil}}{n\cdot\frac{\gamma}{\gamma + 2}+1} & = \frac{\gamma + 4 - \frac{2(\gamma + 2)}{n} - \frac{\gamma + 2}{\lceil n(\gamma + 2)\rceil}}{\gamma + \frac{\gamma + 2}{n}} \\
    & \geq \frac{4 + \gamma - \frac{8}{n} - \frac{\gamma + 2}{n(\gamma + 2)}}{\gamma + \frac{4}{n}} = \frac{4 + \gamma - \frac{9}{n}}{\gamma + \frac{4}{n}}. \qedhere
\end{align*}
\end{proof}

\subsection{Missing proofs from Section~\ref{subsec:PFAlinear} (homogeneous linear functions)}\label{appendix:linear}

\thmLinear*

We will first need to prove a series of supporting lemmas. Let us define the family of instances $S_L(n, k, t, z) = \{\{(t, 0), (t+1, 0)\}^k\uplus\{(t,0),(t+1,z)\}^{n-k}\}$ for $n, k, t, \in \mathbb{N}, k\leq n$ and observe that $(S_L(n,k,t,z))_C = S(n,k,t,z)$ as defined in the constant function setting. 

\begin{lemma} \label{lem:zerooptlin}
    For any $n, k, t \in \mathbb{N}$ such that $0 < k \leq n$ and $t \geq n$, and for any $z \in \mathbb{R}$, the optimal homogeneous linear function for instance $S_L(n, k, t, z)$ is $f_{L,0}$.
\end{lemma}
\begin{proof}
    By Lemma~\ref{lem:zeroopt}, we know that constant function $f_{0}$ is optimal for instance $S(n, k, t, z)$. Applying Corollary~\ref{cor:sameminimizer} completes the proof.
\end{proof}

\begin{lemma}\label{lem:singlepeaklin}
    For any $n, k, t \in \mathbb{N}$ and any $S \in S_L(n, k, t, z)$, let $R(f, S)$ denote the absolute risk of homogeneous linear function $f$ with respect to $S$. Then if $f_{L,a}$ is the unique optimal homogeneous linear function over $S$, then for any $b, c$ such that $a < b < c$ or $a > b > c$, the following holds:
    \[R(f_{L,a}, S) < R(f_{L,b}, S) < R(f_{L,c}, S).\]
\end{lemma}
\begin{proof}
Consider $S_C$, which by Corollary~\ref{cor:sameminimizer} has unique optimal constant function $f_a$. Combining Lemma~\ref{lem:consttolin} which states that $R(f_{L,a'}, S) = \frac{|S_C|}{|S|}R(f_{a'}, S_C)$ and Observation~\ref{obs:singlepeak} completes the proof.
\end{proof}

\begin{lemma} \label{lem:underonelin}
    For any $\gamma \in (0,2]$, fix $n, k\in \mathbb{N}$ such that $k = n\cdot \frac{\gamma}{\gamma + 2}$ and $n > k + 1$ and let $t = \lceil n(\gamma + 2)\rceil$. Then for any deterministic mechanism $\mathcal{M}$ that is $1 + \gamma$ consistent, there exists $\delta = \delta(n, \gamma) > 0$ such that, given instance and correct advice $(S_L(n, k+1, t, 1),f_{L,0})$, if $\mathcal{M}$ returns $f_{L,c}$, $c < 1 - \delta$ must hold.
\end{lemma}
\begin{proof}
 Fix $\gamma \in (0,2]$ and $n,k,t$ that satisfy the conditions in the lemma statement. Observe by Lemma~\ref{lem:zerooptlin} that the optimal homogeneous linear function for $S_L = S_L(n, k+1, t, 1)$ is $f_{L,0}$, as $k + 1 > 0$ and $t \geq n$ for any $\gamma$. Then to achieve $1 + \gamma$ consistency, given instance $S_L$ and correct advice $f_{L,0}$, $\mathcal{M}$ must return some $f_{L,c}$ such that $\frac{R(f_{L,c}, S_L)}{R(f_{L,0}, S_L)} \leq 1+\gamma$. Applying Lemma~\ref{cor:sameratio}, we know that 
 \[\frac{R(f_{L,c}, S_L(n, k+1, t, 1))}{R(f_{L,0}, S_L(n, k+1, t, 1))} = \frac{R(f_c, S(n, k+1, t, 1))}{R(f_0, S(n, k+1, t, 1))}.\]

We showed in the proof of Lemma~\ref{lem:underone} that, for $c \in [0,1]$, 
\[\frac{R(f_c, S(n, k+1, t, 1))}{R(f_0, S(n, k+1, t, 1))} \geq 1 + c\left[ \gamma + \frac{3}{2n - \gamma - 2}\right].\]
 Then for $\frac{R(f_{L,c}, S_L(n, k+1, t, 1))}{R(f_{L,0}, S_L(n, k+1, t, 1))}$ to be bounded above by $1 + \gamma$, $c < 1$ must hold, so there exists $\delta > 0$ such that $\mathcal{M}$ returning $f_{L,c}$ must have $c < 1-\delta$ to satisfy $1+\gamma$ consistency.
\end{proof}

\begin{corollary}\label{cor:scalelin}
    For any $\gamma \in (0,2]$, fix $n, k\in \mathbb{N}$ such that $k = n\cdot \frac{\gamma}{\gamma + 2}$ and $n > k + 1$ (note such always exist), and let  $t = \lceil n(\gamma + 2)\rceil$. Then for any deterministic mechanism $\mathcal{M}$ that is $1 + \gamma$ consistent, there exists $\delta = \delta(n, \gamma) > 0$ such that, for all $z > 0$, given instance and correct advice $(S_L(n, k+1, t, z), f_{L,0})$, $\mathcal{M}$ must return $f_{L,c}$ such that $c < (1 - \delta)\cdot z$.
\end{corollary}
The proof is exactly as above in Lemma~\ref{lem:underonelin}, scaling all risk values by $z$, and we also note that this means $\delta(n, \gamma)$ is also the same as in Lemma~\ref{lem:underonelin}.

\begin{lemma}\label{lem:largeinstlin}
    For any $\gamma \in (0,2]$, fix $n, k \in \mathbb{N}$ such that $k = n \cdot \frac{\gamma}{\gamma+2}$ and $n > k+1$, and let $t = \lceil n(\gamma + 2)\rceil$. Fixing $D > 1$, for any strategyproof deterministic mechanism $\mathcal{M}$ that is $1 + \gamma$ consistent, given instance and correct advice $(S_L(n, k+1, t, D'), f_{L,0})$ for any $D' \geq D$, $\mathcal{M}$ must return $f_{L,c}$ such that $c < 1$.
\end{lemma}

\begin{proof}
    Take $\delta$ as in Lemma~\ref{lem:underonelin} and fix some $D > 1$. It suffices to show that for some $D'$ such that $1 + D' \cdot \frac{\delta}{2t+1} \geq D$, $\mathcal{M}$ must return some constant less than $1 - \delta$ given instance $(S_L(n,k+1,t,1 + D'\cdot\frac{\delta}{2t+1}),f_{L,0})$. For simplicity, let us denote $z(i) = 1+ i \cdot \frac{\delta}{2t+1}$ and $S_L(n, k+1, t, z(i))$ as $S_L^i$. We will show by induction that given $(S_L^i, f_{L,0})$ for any $i \leq D'$, $\mathcal{M}$ must return $f_{L,c}$ such that $c < 1 - \delta$.

    Observe that for $i = 0$, this holds by Lemma~\ref{lem:underonelin}. For any $i \leq D'$, we know by Lemma~\ref{lem:zerooptlin} that $f_{L,0}$ is the optimal solution to $S_L^i$. Then it is sufficient to prove the following claim.
    
    \paragraph{Claim:} For any $i < D'$, if $\mathcal{M}$ must return some $f_{L,c}$ such that $c < 1 - \delta$ given instance $S_L^i$ and advice $f_{L,0}$, then given instance $S_L^{i+1}$ and advice $f_{L,0}$, $\mathcal{M}$ must still return some $f_{L,c'}$ such that $c' < 1 - \delta$.

    Fix $i < D'$. Let us define for any $j \leq n-k-1$ the instance $S_L^i(j)$ as 
    \[\{\{(t, 0), (t+1, 0)\}^{k+1}\uplus \{(t, 0), (t+1, z(i+1))\}^j \uplus \{(t, 0), (t+1, z(i))\}^{n-k-j-1}\}\]
    and observe that $S_L^i(j)' = S^i(j)$ from Lemma~\ref{lem:largeinst}. Since $f_0$ is the optimal constant function for $S^i(j)$, then $f_{L,0}$ is optimal for $S_L^i(j)$ by Corollary~\ref{cor:sameminimizer}. Then we will prove the following: for any $j < n-k-1$, if given instance and advice $(S_L^i(j), f_{L,0})$ $\mathcal{M}$ must return some $f_{L,c}$ such that $c < 1 - \delta$, then given $(S_L^i(j+1),f_{L,0})$, it must still return some $f_{L,c'}$ such that $c' < 1 - \delta$.

Observe that $S_L^i(0) = S_L^i$, so from the claim statement we have that $\mathcal{M}$ returns some $f_{L,c}$ such that $c < 1 - \delta$ given $(S_L^i(0), f_{L,0})$. Assuming then that $\mathcal{M}$ returns some $f_{L,c}$ such that $c < 1 - \delta$ given $(S_L^i(j), f_{L,0})$ for some $j < n-k-1$, we want to show that given $S_L^i(j+1)$ mechanism $\mathcal{M}$ must still return $f_{L,c'}$ for some $c' < 1 - \delta$. First, we will show that given instance $(S_L^i(j+1), f_{L,0})$, $\mathcal{M}$ cannot return $f_{L,c'}$ such that $c' \geq z(i+1)$ while maintaining $1 + \gamma$ consistency. We show in the proof of Lemma~\ref{lem:largeinst} that for any such $c$, 
\[\frac{R(f_{z(i+1)}, S^i(j+1))}{R(f_0, S^i(j+1))} > 1+\gamma.\]
Then since $\frac{R(f_{L, z(i+1)}, S_L^i(j+1))}{R(f_{L,0}, S_L^i(j+1))}$ is equal to this ratio by Corollary~\ref{cor:sameratio}, we know that $\mathcal{M}$ returning some $f_{c'}$ such that $c' \geq z(i+1)$ violates $1 + \gamma$ consistency.

    We can then show that given $(S_L^i(j+1), f_{L,0})$, $\mathcal{M}$ may not return $f_{L,c'}$ such that $c' \in [1-\delta, z(i)]$. To see this, consider that instances $S_L^i(j)$ and $S_L^i(j+1)$ that only differ by the report of one agent denoted $\bar{a}$, switching from type $S_{L,\bar{a}} = \{(t, 0), (t+1, z(i))\}$ to type $\bar{S}_{L, a} = \{(t, 0), (t+1, z(i+1))\}$. Observe by Lemma~\ref{lem:singlepeaklin} that $R_{\bar{a}}(f_{L, c'}, S_{L,\bar{a}})$ increases as $c'$ decreases from $z(i)$ since $f_{L, z(i)}$ is the personal risk minimizing homogeneous linear function for $S_{L,\bar{a}}$. Since the induction hypothesis (on $j$) states that given $(S_L^i(j), f_{L,0})$, mechanism $\mathcal{M}$ returns some $f_{L,c}$ such that $c < 1 - \delta$, returning $c' \in [1 - \delta, z(i)]$ when the instance switches to $S_L^i(j+1)$ would decrease the risk for agent $\bar{a}$, and therefore if $\bar{a}$ had true type $S_{L,\bar{a}}$ they would have incentive to misreport their type as $\bar{S}_{L,\bar{a}}$, violating strategyproofness of $\mathcal{M}$.

    Now we consider the case when $\mathcal{M}$ returns some $c' \in (z(i), z(i+1))$. Observe by Lemma~\ref{lem:singlepeaklin} that $R_{\bar{a}}(f_{c'}, S_{L,\bar{a}})$ increases as $c'$ increases from $z(i)$, so $R_{\bar{a}}(f_{L, c'}, S_{L,\bar{a}}) < R_{\bar{a}}(f_{L, z(i+1)}, S_{L,\bar{a}})$. We want to compare this with $R_{\bar{a}}(f_{L,1-\delta}, S_{L,\bar{a}})$. We know by Lemma~\ref{lem:consttolin} that
    \[R_{\bar{a}}(f_{L, 1-\delta}, S_{L,\bar{a}}) - R_{\bar{a}}(f_{L, z(i+1)}, S_{L,\bar{a}}) = \frac{|S|}{|S_L|}\left[R_{\bar{a}}(f_{1-\delta}, S_{\bar{a}}) - R_{\bar{a}}(f_{z(i+1)}, S_{\bar{a}})\right],\]
which we know to be nonnegative for all $i \geq 0$ because $R_{\bar{a}}(f_{1-\delta}, S_{\bar{a}}) - R_{\bar{a}}(f_{z(i+1)}, S_{\bar{a}})$ was shown to be nonnegative in the proof of Lemma~\ref{lem:largeinst}. Then $R_{\bar{a}}(f_{L, c'}, S_{L,\bar{a}}) < R_{\bar{a}}(f_{L, z(i+1)}, S_{L,\bar{a}}) \leq R_{\bar{a}}(f_{L,1-\delta}, S_{L,\bar{a}})$, so if $\mathcal{M}$ returns $f_{L,c'}$ such that $c' \in (z(i), z(i+1))$ given $(S_L^i(j+1), f_{L,0})$, agent $a$ with true type $S_{L,\bar{a}}$ would have incentive to misreport their type as $\bar{S}_{L,\bar{a}}$ and decrease their personal risk. Then, $\mathcal{M}$ we have that $c' < 1 - \delta$ must hold. 

Setting $j = n - k - 2$, we get that $\mathcal{M}$ must return $f_{L,c}$ such that $c <1 - \delta$ given instance $S_L^i(n - k - 1) = S_L^{i+1}$, thus proving the claim. Then the induction on $i$ holds for any $i \leq D'$, and $\mathcal{M}$ must return $f_{L,c}$ for some $c < 1 - \delta$ given instance $S^{D'}$ and advice $f_{L,0}$.
\end{proof}

Now we can prove the main theorem.

\begin{proof}[Proof of Theorem~\ref{thm:linear}]

Taking $D = \frac{2}{\varepsilon} + \frac{8}{\gamma \varepsilon}$, we know by Lemma~\ref{lem:robustfunction} that there exist $n, k \in \mathbb{N}$ such that $k = n\cdot \frac{\gamma}{\gamma + 2}, n > k+1$, and $r(n, D) > 1 + 4/\gamma - \varepsilon$. Observe that $T > D$ by definition of $T$. Letting $t = \lceil n(\gamma + 2)\rceil$, 
we know by Lemma~\ref{lem:largeinstlin} that $\mathcal{M}$ must return $f_{L,c}$ for $c < 1$ given instance and correct advice $(S_L(n,k+1,t,D'),f_{L,0})$ for any $D' \geq D$, so we consider instance $S_L(n, k+1, t, T)$. 

    Let us assign indices to the agents of $S_L^0$:
    \[S_{L,1}^0 = \ldots = S_{L,k+1}^0 = \{(t,0), (t+1, 0)\}\text{ and }S_{L,k+2}^0 = \ldots = S_{L,n}^0 = \{(t, 0), (t+1, T\cdot(t+1))\}.\]

We now construct a series of instances for which we will prove that $\mathcal{M}$ must return $f_{L,c}$ such that $c < 1$ in order to satisfy strategyproofness. We define agent types $S_{L,j}^i$, for $i \in [0, n]$ and $j\in[n]$. 
\[\text{If }0 \leq i \leq k+1, \; S_{L,j}^i = \begin{cases}
    \{(t, T), (t+1, 0)\} & j \leq i \\
    \{(t, 0), (t+1, 0)\} & i < j \leq k+1 \\
    \{(t, 0), (t+1, T\cdot(t+1))\} & k +1< j \leq n
    \end{cases},\]
\[\text{else if }k+1 < i \leq n, \; S_{L,j}^i = \begin{cases}
        \{(t, T\cdot t), (t+1, 0)\} & j \leq k+1 \\
        \{(t, T\cdot 5), (t+1, T\cdot(t+1))\} & k+1 < j \leq i \\
        \{(t, 0), (t+1, T\cdot(t+1))\} & i < j \leq n
    \end{cases}.
\]
Let each instance $S_L^i = \biguplus_{l = 1}^nS_{L,l}^i$, which is consistent with the original definition of $S_L^0$, and observe that for every $i \in [n]$, $S_L^{i-1}$ and $S_L^i$ that only differ by the ith agent's type. If we let $S_{L,-j}^i = \biguplus_{l \in [n]\backslash j}S_{L,l}^i$, we can see that $S_{L, -i}^{i-1} = S_{L, -i}^i$. 

Our induction hypothesis is that, given $(S_L^i, f_{L,0})$, mechanism $\mathcal{M}$ must return $f_{L,c}$ such that $c < 1$, which we know holds for $i = 0$. Now assume that this holds for $i - 1$ and $\mathcal{M}(S_L^{i-1}, f_{L,0}) = f_{L,c}$ for some $c < 1$. We first consider when $1 \leq i \leq k + 1$. Recall that 
\[S_{L,i}^{i-1} = \{(t, 0), (t+1, 0)\}\text{ and }S_{L,i}^i = \{(t, T\cdot t)), (t+1, 0)\}.\]
Both of these have unique personal risk minimizer $f_{L,0}$. Assume for contradiction that $\mathcal{M}(S_L^i, f_{L,0}) = f_{L,b}$ for some $b \geq 1$. Then 
\[R_i(\mathcal{M}(S_L^{i-1}, f_{L,0}), S_{L,i}^i) = R_i(f_{L,c}, S_{L,i}^i) < R_i(f_{L,b}, S_{L,i}^i) = R_i(\mathcal{M}(S_L^i, f_{L,0}), S_{L,i}^i) ,\]
where the inequality is by Lemma~\ref{lem:singlepeaklin} since $0 < c < b$.
Since $S_L^{i-1} = S_{L,i}^{i-1} \uplus S_{L,-i}^i$ and $S_L^i = S_{L,i}^i \uplus S_{L,-i}^i$, we get that 
\[R_i(\mathcal{M}(S_{L,i}^{i-1} \uplus S_{L,-i}^i, f_{L,0}), S_{L,i}^i) < R_i(\mathcal{M}(S_{L,i}^i \uplus S_{L,-i}^i, f_{L,0}), S_{L,i}^i).\]
This means, fixing the types of all other agents to $S_{L,-i}^i$, an agent with type $S_{L,i}^i$ can misreport their type as $S_{L,i}^{i-1}$ and decrease their personal risk, which violates strategyproofness. Then $\mathcal{M}$~cannot return $f_{L,b}$ for any $b \geq 1$ given instance $(S_L^i, f_{L,0})$. Thus the claim holds for $i = k + 1$.

Now we consider $k + 1 < i \leq n$, and again assume that $\mathcal{M}(S_L^{i-1}, f_{L,0}) = f_{L,c}$ for some $c < 1$. Recall that 
\[S_{L,i}^{i-1} = \{(t, 0), (t+1, T\cdot(t+1))\} \text{ and }S_{L,i}^i = \{(t, T\cdot t), (t+1, T\cdot(t+1)\}.\]
Both of these have unique personal risk minimizer $f_{L,T}$. Assume for contradiction that $\mathcal{M}(S_L^i, f_{L,0}) = f_{L,b}$ for some $b \geq 1$. Then 
\[R_i(\mathcal{M}(S_L^{i-1}, f_{L,0}), S_{L,i}^i) = R_i(f_{L,c}, S_{L,i}^i) > R_i(f_{L,b}, S_{L,i}^i) = R_i(\mathcal{M}(S_L^i, f_{L,0}), S_{L,i}^i),\]
again by Lemma~\ref{lem:singlepeaklin} because $T > b > c$, and $\mathcal{M}$ cannot output $f_{L,a}$ for $a > T$. Since $S_L^{i-1} = S_{L,i}^{i-1} \uplus S_{L,-i}^i$ and $S_L^i = S_{L,i}^i \uplus S_{L,-i}^i$, we get that 
\[R_i(\mathcal{M}(S_{L,i}^{i-1} \uplus S_{L,-i}^i, f_{L,0}), S_{L,i}^i) > R_i(\mathcal{M}(S_{L,i}^i \uplus S_{L,-i}^i, f_{L,0}), S_{L,i}^i).\]
Then, if the types of all other agents is $S_{L,-i}^i$, an agent with type $S_{L, i-1}^i$ can misreport their type as $S_{L,i}^i$ and decrease their personal risk, again violating strategyproofness. Then $\mathcal{M}$ cannot return $f_{L,b}$ having seen instance $(S_L^i, f_{L,0})$. Thus the claim holds for $i = n$. 

We can now bound the robustness of $\mathcal{M}$ by considering its performance on the instance $(S_L^n, f_{L,0})$. Recall that $S_L^n$ is defined by agent types
\[S_{L,1}^n = \dots = S_{L,k+1}^n = \{(t, T), (t+1, 0)\}\text{ and }S_{L,k+2}^n = \dots = S_{L,n}^n = \{(t, T\cdot t), (t+1, T\cdot(t+1))\}.\]
Since we know $\mathcal{M}$ returns $f_{L,c}$ for some $c < 1$, we can compare the incurred risk with that of returning $f_{L,T}$. By Corollary~\ref{cor:sameratio} we know that 
\[\frac{R(f_{L,c}, S_L^n)}{R(f_{L,T}, S_L^n)} = \frac{R(f_c, S^n)}{R(f_{T}, S^n)}\]
for $S^n = (S_L^n)_C$, which was shown in the proof of Theorem~\ref{thm:rationalconstantlb} to be at most $r(n, D) > \beta$, thus $\mathcal{M}$ cannot be $\beta$ robust. \end{proof}

\section{Missing proofs from Section~\ref{sec:classification}} \label{appendix:classification}
\subsection{Missing proofs from Section~\ref{subsec:genclassdet} (general decision)}\label{appendix:genclassdet}

\thmRandClass*

We first define certain types of classification mechanisms using definitions from voting theory. A mechanism is a \emph{lottery} over mechanisms $\mathcal{M}_1,\mathcal{M}_2,\dots$ if it runs mechanism $\mathcal{M}_j$ with probability $p_j$. Then we care about three types of mechanisms. A \emph{random dictator mechanism} is a lottery over dictatorial mechanisms, a \emph{duple mechanism assigns probability 0 to all labelings except (at most) two}, and a \emph{random dictator duple mechanism} (RDD) is a lottery over random dictator and duple mechanisms. We will leverage the following characterization of strategyproof random mechanisms.

\begin{lemma}[\citealt{meir2011tight}]\label{lem:MisRDD}
    Let $\mathcal{M}$ be a strategyproof randomized mechanism that is $\mu$-granular for some $\mu \in (0,1]$. Then for any $k \geq 2/\mu$, $\mathcal{M}$ is a RDD over $(X, \mathcal{C}, \mathcal{S})(k)$.
\end{lemma}

Now we can prove our result.
\begin{proof}[Proof of Theorem~\ref{thm:randomclassification}]

 Let us first introduce the family of instances we focus on. For fixed $k \in \mathbb{N}$, consider the input space over $m=3k$ points \[X(k)=\{\underbrace{x,\ldots,x}_k,\underbrace{y,\ldots,y}_k,\underbrace{z,\ldots,z}_k\}.\]
Note that the identities of $x,y,z$ do not matter and they are just for ease of presentation. Then we consider the set of labelings $\mathcal{C}(k) = \{c_x, c_y, c_z\}$ where, for $w \in \{x, y, z\}$, $c_w$ labels any $w' \in X$ such that $w' = w$ with 1 and all others with 0. For example, $c_x = \{\underbrace{1,\ldots,1}_k,\underbrace{0,\ldots,0}_k,\underbrace{0,\ldots,0}_k\}$. Let us define a family of instances $\mathcal{S}(k)$ such that every agent has type defined as follows: one of the point types $\{x, y, z\}$ is labeled as all 0's, the second as all 1's, and the third has at least one label each of 0 and 1. For example,  $\{\underbrace{1,\ldots,1}_k,\underbrace{0,\ldots,0}_k,\underbrace{1,0,\ldots,0}_k\}$. We will denote the tuple $(X(k), \mathcal{C}(k), \mathcal{S}(k))$ as $(X, \mathcal{C}, \mathcal{S})(k)$.

      Observe that, for any prediction $\Tilde{c} \in \mathcal{C}(k)$, the mechanism $\mathcal{M}_{\Tilde{c}}$ induced by fixing the prediction is a strategyproof randomized mechanism without advice, and it is $\mu$-granular. WLOG, let $\Tilde{c} = c_x$ be the advice. Then by Lemma~\ref{lem:MisRDD}, $\mathcal{M}_{\Tilde{c}}$ must be an RDD over $(X, \mathcal{C}, \mathcal{S})(k)$ for $k \geq 2/\mu$. As a result, there must exist duple mechanisms $D_{x,y}$, $D_{y,z}$, and $D_{x,z}$ (where $D_{w, w'}$ denotes a duple mechanism over $c_w$ and $c_{w'}$), as well as random dictator mechanism $RD$, such that we can express the expected risk of $\mathcal{M}_{\Tilde{c}}$ as follows: 
   \begin{align*}
       R(\mathcal{M}_{c_x}(S), S) = & p_{x,y}\cdot R(D_{x, y}(S), S) + p_{y,z}\cdot R(D_{y,z}(S), S) + p_{x,z}\cdot R(D_{x,z}(S), S)\\
       & + (1 - p_{x,y} - p_{y,z} - p_{x,z})\cdot R(RD(S), S),
   \end{align*}
   where $p_{x,y}, p_{y,z}, p_{x,z} \geq 0$ and $p_{x,y}+p_{y,z}+p_{x,z} \leq 1$.

   Let us first bound the consistency of $\mathcal{M}$, or the approximation ratio of $\mathcal{M}_{c_x}$ over instances where $c_x$ is the optimal labeling. To do so, we will lower bound the approximation of any random dictator mechanism over $n$ agents. Observe that any random dictator mechanism must pick one of the agents with probability $\sigma \geq \frac{1}{n}$, and WLOG let that agent be agent 1. Then consider the following instance. The first agent has type 
   \[S_1 = \{\underbrace{0,1,\ldots,1}_k,\underbrace{1,\ldots,1}_k,\underbrace{0,0,\ldots,0}_k\}\]
   and all other agents $i \in [2, n]$ have type
   \[S_i = \{\underbrace{1,\ldots,1}_k,\underbrace{1,0,\ldots,0}_k,\underbrace{0,0,\ldots,0}_k\}.\] Note that $S \in \mathcal{S}(k)$.
   We want to analyze the performance of any random dictator mechanism on this instance.
   To compare the risks incurred by $c_x$ and $c_y$, we can perform the following calculations.
   \begin{align*}R(c_x, S) & = \sigma\cdot R_1(c_x, S_1) + (1 - \sigma)\cdot R_i(c_x, S_i) \\
   & = \frac{1}{3k}\left[\sigma\cdot (k+1) + (1 - \sigma)\right] \geq \frac{1}{3k}\left[\frac{1}{n}(k+1) + \frac{n-1}{n}\right]\\
   R(c_y, S)& = \sigma\cdot R_1(c_y, S_1) + (1 - \sigma)\cdot R_i(c_y, S_i) \\
   & = \frac{1}{3k}\left[\sigma\cdot (k-1) + (1 - \sigma)\cdot(2k+1)\right] \leq \frac{1}{3k}\left[\frac{1}{n}(k-1) + \frac{n-1}{n}(2k+1)\right].\end{align*}
Observe that for any $k \geq 3$, we have that $R(c_x, S) < R(c_y, S)$, so the prediction $c_x$ is correct. \citet{meir2011tight} showed that the approximation ratio of this instance is at least $3 - \frac{2}{n} - \frac{2n + 96}{k}$, so the consistency of $\mathcal{M}$ is at least 
   \[(1 - p_{x,y} - p_{y,z} - p_{x,z})\cdot \frac{R(RD(S),S)}{R(c_x, S)} \geq (1 - p_{x,y} - p_{y,z} - p_{x,z})\left(3 - \frac{2}{n} - \frac{2n + 96}{k}\right),\]
   and this value is bounded below by $3 - \frac{2}{n} - \frac{2n + 96}{k} - 3\alpha$ for $\alpha = p_{x,y} + p_{y,z} + p_{x,z}$.

   We will now bound the robustness of $\mathcal{M}$. We will do so by lower bounding the approximation of duple mechanisms. For any $w, w', w'' \in \{x, y, z\}$, we will consider an instance where every agent vastly prefers $c_w$ over all labelings and slightly prefers $w'$ over $w''$. WLOG, let $(w, w', w'') = (x, y, z)$. Then in instance $S$ all agents will have type 
   \[S_i = \{\underbrace{1,\ldots,1}_k,\underbrace{1,0,\ldots,0}_k,\underbrace{0,0,\ldots,0}_k\}. \]
   Again observe that $S \in \mathcal{S}(k)$. Consider duple mechanism $D_{y,z}$. Since it can only pick amongst $c_y$ and $c_z$, it has approximation ratio at least
   \[\frac{R(D_{y,z}(S), S)}{R(c_x, S)} \geq \frac{\min\{R(c_y, S), R(c_z, S)\}}{R(c_x, S)} = \frac{R(c_y, S)}{R(c_x, S)} = \frac{k+1}{1} = k+1.\]
   Then for each $w, w' \in \{x, y, z\}$, there exists $S \in \mathcal{S}(k)$ such that $\frac{R(\mathcal{M}_{\Tilde{c}}(S), S)}{R(c^*, S)} \geq p_{w, w'}\cdot (k+1)$. This means that $\mathcal{M}$ has robustness at least 
   \[\max\{p_{x,y}, p_{y,z}, p_{x,z}\}\cdot (k+1) \geq \frac{p_{x,y} + p_{y,z} + p_{x, z}}{3}(k+1) = \frac{\alpha}{3}(k+1).\]

   In order for $\mathcal{M}$ to be $3 - \frac{2}{n} - \varepsilon$ consistent, we need $3 - \frac{2}{n} - \frac{2n+96}{k} - 3\alpha \leq 3 - \frac{2}{n} - \varepsilon$, and selecting $k  \geq \frac{4n + 192}{\varepsilon}$ means that $\alpha > \frac{\varepsilon}{6}$ must hold. Then $\mathcal{M}$ has robustness at least 
   \[\frac{\alpha}{3}(k+1) > \frac{\varepsilon}{18}(k+1) = O(n). \qedhere
   \]
\end{proof}

\subsection{Missing proofs and results from Section~\ref{subsec:binclass} (two labelings)} \label{appendix:binary}

\thmToBinary*

\begin{proof}
Consider set of labelings $\mathcal{C} = \{c',c''\}$ and instance-advice pair $(S,\Tilde{c})$ for some $\Tilde{c} \in \mathcal{C}$.  Let us denote the set of indices on which $c'$ and $c''$ differ as $J = \{j \in [m]: c'(x_j) \neq c''(x_j)\}$. Observe that for any $j\notin J$, meaning $c'(x_j) = c''(x_j)$, $\mathcal{M}$ must label $x_j$ with that same value, regardless of whether $c'$ or $c''$ is selected. Then the relative error over points $x_j$ for $j \notin J$ is the same for both labelings, that is 
    \[\sum_{i=1}^n\sum_{j \notin J}\mathbf{1}(c'(x_j) \neq y_{i,j}) = \sum_{i=1}^n\sum_{j \notin J}\mathbf{1}(c''(x_j) \neq y_{i,j}) \]
whose value we will denote $Err_{[m] \backslash J}$.

Let us define our dataset restricted to points with indices in $J$ using $X_J = \{x_j: j\in J\}$, and similarly $Y_{i,J} = \{y_{i,j}: j \in J\}$ for each $i \in [n]$, so that $Y_J = \uplus_{i = 1}^n Y_{i,J}$ and $S_J = \uplus_{i=1}^nS_{i,J} = \uplus_{i=1}^n(X_J,Y_{i,J})$. We can then create for each agent $i$ transformed data $Y^T_i$ labeling points in $X_J$ composed of the following:
\[y^T_{i,j} = \begin{cases}
    0 & \text{if }y_{i,j} = c'(x_j) \\
    1 & \text{if }y_{i,j} = c''(x_j),
\end{cases}\]
and let $Y^T= \uplus_{i=1}^nY^T_i$, so $S^T = (X_J, Y^T)$ is our transformed dataset. Our transformed advice will be denoted $\Tilde{c}^T$, which equals $c_0$ if $\Tilde{c} = c'$ and $c_1$ if $\Tilde{c} = c''$. Note that $\Tilde{c}^T$ is correct over $S^T$ amongst options $\{c_0,c_1\}$ if $\Tilde{c}$ is correct over $S$ amongst $\{c',c''\}$. Then let $\mathcal{M}'(S,\Tilde{c})$ return $c'$ if $\mathcal{M}(S^T,\Tilde{c}^T) = c_0$ and $c''$ if $\mathcal{M}(S^T,\Tilde{c}^T) = c_1$. Since no randomized is introduced when mapping $S$ to $S^T$ and $\Tilde{c}$ to $\Tilde{c}^T$, nor is it introduced when determining the outcome of $\mathcal{M}'$ given $\mathcal{M}(S^T,\Tilde{c}^T)$, then $\mathcal{M}'$ is deterministic if $\mathcal{M}$ is.

We can now analyze the performance of $\mathcal{M}'$. Abusing notation, let $c'$ and $c''$ also be labelings over $X_J$. Then it is easy to see that for all $i$, $Err(c_0, S^T_i) = Err(c', S_{i,J})$, and similarly $Err(c_1, S^T_i) = Err(c'', S_{i,J})$. Then the efficiency of $\mathcal{M}$ upper bounds the efficiency of $\mathcal{M}'$:

\begin{align*}
    \frac{R(\mathcal{M}'(S,\Tilde{c}), S)}{\min_{c^* \in \{c',c''\}}R(c^*, S)} & = \frac{\sum_{i=1}^n\mathbb{E}[Err(\mathcal{M}'(S,\Tilde{c}),S_{i,J})]+\mathbb{E}[Err_{[m]\backslash J}]}{\min_{c^* \in \{c',c''\}}\sum_{i=1}^n\mathbb{E}[Err(c^*,S_{i,J})] + \mathbb{E}[Err_{[m]\backslash J}}]\\
    & \leq \frac{\sum_{i=1}^n\mathbb{E}[Err(\mathcal{M}'(S,\Tilde{c}),S_{i,J})]}{\min_{c^* \in \{c',c''\}}\sum_{i=1}^n\mathbb{E}[Err(c^*,S_{i,J})]} \\
    & = \frac{\sum_{i=1}^n\mathbb{E}[Err(\mathcal{M}(S^T,\Tilde{c}^T),S^T_i)]}{\min_{c^*\in\{c_0,c_1\}}\sum_{i=1}^n\mathbb{E}[Err(c^*,S^T_i)]}.
\end{align*}
This value is bound above by $\alpha$ when $\Tilde{c}$ is correct by the consistency of $\mathcal{M}$, and by $\beta$ for all $\Tilde{c}$ by the robustness of $\mathcal{M}$.

To show that $\mathcal{M}'$ is strategyproofness when $\mathcal{M}$ is, observe that each agent's reports can only affect labels of $x_j$ for $j \in J$ and the error incurred on these points under $\mathcal{M}'$ is equal to the error incurred on their respective dataset in $S^T$ under $\mathcal{M}$. Then if $\mathcal{M}'$ were not strategyproof, where would be a way for an agent to misreport to decrease their risk in $S^T$ under the outcome of $\mathcal{M}$, contradicting that $\mathcal{M}$ is strategyproof.
\end{proof}

\end{document}